\documentclass[journal]{IEEEtran}
\usepackage{cite}
\usepackage[final]{graphicx}
\usepackage{amsmath}
\usepackage{amssymb}
\usepackage[linesnumbered,ruled,vlined]{algorithm2e}
\usepackage{amsfonts}
\usepackage{array}
\usepackage{booktabs}
\usepackage{stfloats}
\usepackage{lineno}
\usepackage{bm}
\usepackage{stfloats}
\usepackage{epstopdf}
\usepackage{amsthm}
\usepackage{color}
\usepackage{multirow}
\usepackage{multicol}
\usepackage{lipsum}
\newcolumntype{C}{>{\centering\arraybackslash$}p{\linewidth}<{$}}

\usepackage[linesnumbered,ruled,vlined]{algorithm2e}
\theoremstyle{plain}
\newtheorem{theorem}{Theorem}

\newtheorem{proposition}{Proposition}

\begin{document}
	\title{Green Interference Based Symbiotic Security in Integrated Satellite-terrestrial Communications}
	\author{Zhisheng~Yin,
		~\IEEEmembership{Member~IEEE},
		Nan~Cheng,~\IEEEmembership{Member~IEEE},
Tom H.~Luan,~\IEEEmembership{Senior Member~IEEE},
Yilong~Hui,~\IEEEmembership{Member~IEEE},
and Wei~Wang,~\IEEEmembership{Member~IEEE}
		\thanks{
		Z. Yin and Tom H. Luan are with State Key Lab. of ISN and School of Cyber Engineering, Xidian University, Xi'an, 710071, China (e-mail: zsyin@xidian.edu.cn; tom.luan@xidian.edu.cn).
		
			N. Cheng and Y. Hui are with State Key Lab. of ISN and School of Telecommunications Engineering, Xidian University, Xi'an, 710071, China (e-mail: dr.nan.cheng@ieee.org; ylhui@xidian.edu.cn).
			
			W. Wang is with	the College of Electronic Information Engineering,
			Nanjing University of Aeronautics and Astronautics, Nanjing, 211106, China (e-mail: wei\_wang@nuaa.edu.cn). 
				
			Corresponding author: dr.nan.cheng@ieee.org and tom.luan@xidian.edu.cn.
			
			This work was supported in part by the Fundamental Research Funds for the Central Universities of Ministry of Education of China under Grant XJS221501, the National Natural Science Foundation of Shaanxi Province under Grant 2022JQ-602, National Natural Science Foundation of China (No. 62071356 and 6210010668), Natural Science Foundation of Jiangsu Province  (No. BK20200440), and Natural Sciences and Engineering Research Council (NSERC) of Canada.		
	}}%
	\maketitle
	
	\IEEEpeerreviewmaketitle
	\begin{abstract}
In this paper, we investigate secure transmissions in integrated satellite-terrestrial communications and the green interference based symbiotic security scheme is proposed. Particularly, the co-channel interference induced by the spectrum sharing between satellite and terrestrial networks and the inter-beam interference due to frequency reuse among satellite multi-beam serve as the green interference to assist the symbiotic secure transmission, where the secure transmissions of both satellite and terrestrial links are guaranteed simultaneously. Specifically, to realize the symbiotic security, we formulate a problem to maximize the sum secrecy rate of satellite users by cooperatively beamforming optimizing and a constraint of secrecy rate of each terrestrial user is guaranteed. Since the formulated problem is non-convex and intractable, the Taylor expansion and semi-definite relaxation (SDR) are adopted to further reformulate this problem, and the successive convex approximation (SCA) algorithm is designed to solve it. Finally, the tightness of the relaxation is proved. In addition, numerical results verify the efficiency of our proposed approach. 
         
\end{abstract}
\section{Introduction}
In the emerging 6G, integrated satellite-terrestrial networks can provide a scenario of Internet access with expanded coverage and seamless connectivity\cite{You2021,Tataria2021}. The use cases of integrated satellite-terrestrial networks have attracted a lot of attention in applications such as Internet-of-things (IoT), vehicle-to-everything (V2X), holographic communications, etc\cite{Lyu2021,Lin2021b,Liu2021,Hui2021}. With the increasing number of connective devices and abundant service requirements, the spectrum resource of space information networks is increasingly scarce\cite{Xu2021}. To address this issue, the spectrum sharing within satellite-terrestrial networks and frequency reuse among multiple beams of satellite are widely adopted to improve the spectrum utilization\cite{Zhang2019}. However, the co-channel/inter-beam interference caused by the spectrum sharing and frequency reuse degrades the overall communication performance\cite{Yang2017,An2019,Lin2022}, and the security vulnerability is opportunistic to penetrate crossing satellite and terrestrial networks\cite{Bankey2019,Yin2019}. Particularly, due to the openness of wireless channels and the broad broadcasting coverage, the integrated satellite-terrestrial communications are vulnerable to eavesdropping threats and the very large and complex geographical areas give harbour to attackers and eavesdroppers (Eves), which results in serious security issues.	       

To combat the eavesdropping threats, physical layer security has been widely investigated in conventional wireless communications and also actively explored in satellite and aerial networks recently. As a supplement of cryptography protocol in the upper layer, physical layer security approaches are aimed to ensure secrecy capacity by exploiting the randomness difference of wireless channels\cite{Zhang2021}. Particularly, several related works of physical layer security in satellite-terrestrial communications have been reported. To achieve the secure transmissions of satellite communications, additional communication facilities generally work as an assistance to enhance the main channel capacity or degrade the eavesdropping channel capacity. By exploiting terrestrial base stations (BSs) to serve as cooperative relays of satellite communications, an opportunistic user-relay selection criteria is proposed to improve the secrecy performance in hybrid satellite-terrestrial relay networks \cite{Bankey2019}. 
An unmanned aerial vehicle (UAV) relay and an aerial Eve are particularly considered in hybrid satellite-terrestrial networks, a 3D mobile UAV relaying strategy is designed and the secrecy capacity of satellite-terrestrial communications is analyzed in \cite{Sharma2020}. A multi-user cooperation scheme is proposed to improve the secrecy rate of satellite communications, where the inter-user-interference serves as the green interference to degrade the Eves\cite{Yin2019}. In addition, the mutual interference between satellite and terrestrial networks generally degrades the system performance \cite{Lin2021b}, 
however the green interference from terrestrial network can be designed by beamforming (BF) optimization to minimize the total transmission power while guaranteeing the secrecy rate constraint of satellite users and the common quality of service of terrestrial users\cite{Lin2018}. However, aforementioned related works in satellite-terrestrial networks only consider the security of satellite or terrestrial link independently and generally external resources are consumed to assist with the concerned secure link.       

\begin{figure}[!h]
	\centering
	\includegraphics[width=0.5\textwidth]{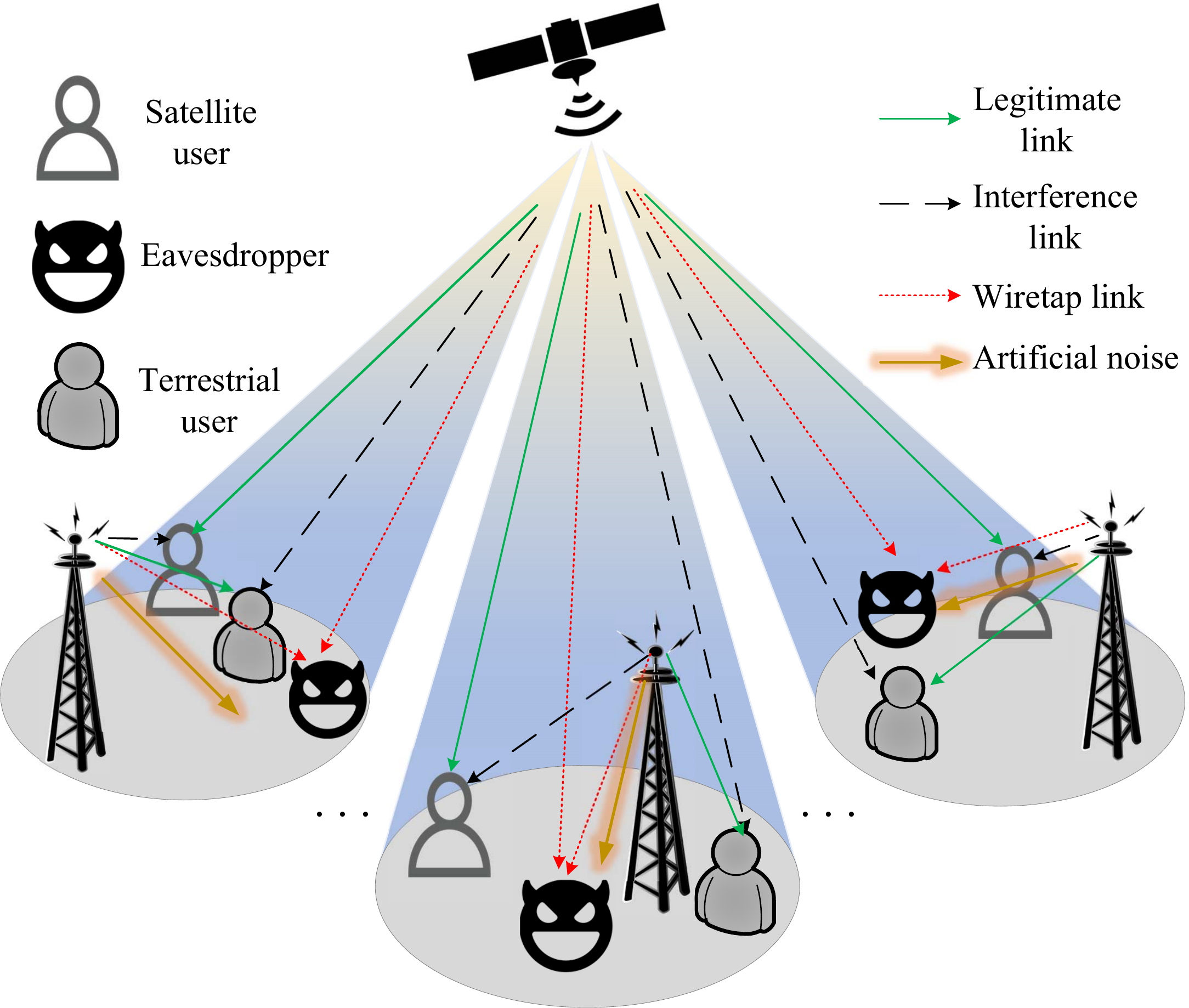}
	\caption{Symbiotic secure transmission model in integrated satellite-terrestrial downlink communications.}
	\label{SM}
\end{figure}
In this paper, we consider the downlink secure transmission in integrated satellite-terrestrial networks as shown in Fig. \ref{SM}, where the multi-beam satellite shares spectrum resource with terrestrial networks and the full frequency reuse is considered among multi-beam. Specifically, a BS coexists in each satellite beam and an Eve is considered to wiretap the satellite and terrestrial link simultaneously. Particularly, the channel similarity between satellite links and limited resource at satellite challenges the implementation of physical layer security in integrated satellite-terrestrial communications. To address this issue and hold the security requirement of terrestrial segment, we conduct a work to guarantee secure transmissions in both satellite and terrestrial communications in this paper and the satellite-terrestrial symbiotic security is realized.     
The main contributions of this paper are as follows.
\begin{itemize}
	\item 
	A framework of satellite-terrestrial symbiotic secure transmission is first proposed, where the green interference from the internal system without additional assistance is considered to implement the symbiotic security, where the co-channel interference caused by the spectrum sharing within satellite-terrestrial networks and the inter-beam interference due to frequency reuse among satellite beams serve as the green interference. 	
	By carefully designing such green interference with BF optimization to confuse the wiretap channel, both satellite and terrestrial links can benefit from each other in achieving the secure transmission. Due to the inherent green interference in integrated satellite-terrestrial communications, the symbiotic security can be realized without consuming external resources.
	\item
	To realize the symbiotic security, we formulate a problem to maximize the sum secrecy rate of satellite users (SUs) in multiple beams by jointly optimizing BFs of satellite and all BSs, and meanwhile a predefined secrecy rate of terrestrial user (TU) in each beam is ensured. To solve this intractable non-convex problem, the Taylor expansion and semi-definite relaxation (SDR) are first adopted to reformulate this problem with constructing convex approximations, and secondly a successive convex approximation (SCA) algorithm is conduct to solve it. 
	\item
	To prove the tightness of relaxation, we introduce an equivalent power minimization problem and by which the rank-one constraints of BF vectors of satellite and BSs are proved. In addition, we develop a reference benchmark which jointly optimizes satellite BF and the power allocation of BSs, where the BS can utilize partial power to create artificial noise (AN) for confusing Eve. Finally, simulations are carried out to verify the efficiency of our proposed approach. 
\end{itemize}

The reminder of this work is organized as follows. In Section III, the system model of symbiotic secure transmission in integrated satellite-terrestrial downlink communications is presented and we formulate the problem to maximize the sum secrecy rate of SUs with constrained secrecy rate of TUs. Then a series of reformations and a cooperative BF optimization are conduct to solve this non-convex problem in Section IV. In Section V, for the reference propose, the joint optimization approach of the satellite BF and the power allocation of BSs is given, and simulation results are carried out to evaluate the secrecy performance of our proposed approaches. Finally, this paper is concluded in Section VI.     

\section{Related Works}
By utilizing multi-antenna technology, spatial modulation based physical layer security approaches can be summarized as precoding, friendly jamming, and transmitting selection\cite{Panayirci2020,Yin2022,Yin2022a}. In terrestrial networks, 
abundant channel uncertainties can be utilized to distinguish with the main and wiretap channels, and recently intelligent reflecting surface (IRS) is introduced to enhance the physical layer security recently\cite{Feng2021,Yang2021}. Particularly, secure BF schemes in multiple-input–single-output (MISO) and multiple-input–multiple-output (MIMO) systems are widely investigated \cite{Zhao2020,Li2020,Tian2020}, and artificial noise (AN) aided secure BF approach is also generally adopted to confuse Eves\cite{Yun2020,Sun2019}. In addition, non-orthogonal multiple access (NOMA) based physical layer security approaches have also attracted some attentions\cite{Liu2017,Xiang2020}. 

Whereas for satellite communications, a two-way physical layer security protocol is proposed in optical links and the random nature of Poisson channel is exploited to ensure secret communications\cite{Hayashi2020}, where a nano-satellite is particularly considered as the Eve in the uplink. A threshold-based
user scheduling scheme is proposed to address a case of multiple eavesdroppers\cite{Guo2020}. 
By jamming from satellite and full-duplex receiver, a joint BF and power allocation scheme is proposed to ensure the secrecy outage probability of multi-beam satellite communications\cite{Cui2020}.  
Particularly, a comprehensive survey on physical layer security in satellite communications is conducted in \cite{Li2020a}. However, the inherent difference between satellite and terrestrial communications has not been clearly pointed out when executing physical layer security respectively. Few works highlight the resource limitation at satellite and the distinctive channel characteristic of satellite channels, i.e., channel similarity\cite{Tropea2021,Yin2022,Yin2022a,PerezNeira2019}. 

Considering a coexistence of satellite and terrestrial networks, current investigations show that either the secure transmission of satellite communication or terrestrial communication is considered. To guarantee secure satellite communications, the signal-to-interference-plus-noise (SINR) of Eve for wiretapping satellite signals is decreased by terrestrial BS's cooperative BF, and the SINR of legitimate satellite user is increased\cite{Du2018}, where the quality of service (QoS) of terrestrial link is guaranteed. The BS provides interference resource to assist secure satellite link in cognitive satellite-terrestrial communications\cite{Yan2020}. A secrecy-energy efficiency is considered while constrains the cellular users' rate requirements in a cognitive satellite-terrestrial network\cite{Lin2021}. 
In addition, by exploiting the BS and a cooperative terminal as green interference resources, the Eve within the signal beampattern region of BS is degraded, where the interference level of satellite link is constrained\cite{Lin2021a}. However, the work considering secure transmissions of both satellite and terrestrial communications has not been investigated, which motivates we conduct an investigation on physical layer security in integrated satellite-terrestrial communications. 

\begin{table}[t]
	\centering
	\caption{Summary of Main Notations and Definitions}
	\label{tab11}
	\begin{tabular}{lll}
		\toprule
		\midrule
		Notation  & Definition \\
		\midrule
		$N$ & number of satellite beams\\
		$M$ & number of BS transmit antennas\\
		${\rm{S}}{{\rm{U}}_{k}}$	&satellite user in the $k^{th}$ beam \\
		${\rm{T}}{{\rm{U}}_{k}}$	&terrestrial user in the $k^{th}$ beam \\
				$ {\bf{v}} \in \mathbb{C}^{(M-1)\times 1}$ & AN vector\\
		$ {\bf{h}}_{su} \in {\mathbb{C}^{N \times 1}}$ & channel vector from satellite to SU\\
		${\bf{h}}_{tu} \in {\mathbb{C}^{N \times 1}}$ & channel vector from satellite to TU\\
		${\bf{h}}_{e} \in {\mathbb{C}^{N \times 1}}$ & channel vector from satellite to Eve\\
		$ {\bf{g}}_{su} \in \mathbb{C}^{M\times 1}$  & channel vector from BS to SU\\
		${\bf{g}}_{tu} \in \mathbb{C}^{M\times 1}$& channel vector from BS to TU\\
		${\bf{g}}_{e} \in \mathbb{C}^{M\times 1}$& channel vector from BS to Eve\\
		$ {{\bf{w}}_k} \in {\mathbb{C}^{N \times 1}} $ & satellite BF vector for the $k^{th}$ beam\\
		$ {\bf{f}}_k \in \mathbb{C}^{M\times 1}$ & BS BF vector in the $k^{th}$ beam\\
		$ \Delta$ & norm-bounded channel estimate error\\
	${\hat h_e}$ & eavesdropping channel estimation by the BS\\
 ${{{\bf{\hat g}}}_e}$ & eavesdropping channel estimation by satellite \\
		\bottomrule
	\end{tabular}
\end{table}

\section{System Model and Problem Formulation}

 We consider the downlink secure transmission in integrated satellite-terrestrial communications shown in Fig. \ref{SM}, where a multi-beam satellite is assumed and a BS with terrestrial network exists in each beam. Full frequency reuse is adopted among satellite beams and the spectrum is shared with the BSs. Particularly, in each beam, we consider a legitimate satellite user (SU) and terrestrial user (TU) coexisting in the common coverage of satellite and the BS. Therefore, when a passive Eve hides in such common area, it could wiretap the SU or TU possibly.
 
 In this section, the channel and signal models of satellite and terrestrial communications are first presented, where the imperfect channel state information (CSI) for the Eve is assumed. Based on the signal models, the received SINRs of SUs, TUs and Eves are given. Further, we conduct a symbiotic security problem that the secure transmissions of both satellite and terrestrial links are guaranteed simultaneously without using additional resources. In addition, the notations are defined in Table I.

\subsection{Channel and Signal Models}
For the satellite-to-ground channel, the free space path loss (FSPL), rain attenuation, and satellite beam gain are generally considered to construct the channel model\cite{Series2015}, which is given by
\begin{equation}\label{key}
{\bf{h}} = \sqrt {{C_L}b\beta } \exp \left( { - j{\bm{\theta }}} \right) ,
\end{equation}
where $ C_L $ denotes the FSPL, $ b $ denotes the beam gain, $ \beta $ denotes the channel gain due to rain attenuation,
and $ \bm{\theta } $ is the phase vector with uniform distribution over $ \left[ {\left. {0,2\pi } \right)} \right.$.
Specifically, 
\begin{equation}\label{key}
	{C_L} = {\left( {{\lambda  \mathord{\left/
					{\vphantom {\lambda  {4\pi }}} \right.
					\kern-\nulldelimiterspace} {4\pi }}} \right)^2}/\left( {{d^2} + {h^2}} \right),
\end{equation}
where $  \lambda$ denotes signal wavelength, $ d $ denotes the distance
from the beam center to the center of satellite coverage, and
$ h $ accounts for the height of satellite. 		
The beam gain is defined by
\begin{equation}\label{bik}
	b = G{\left( {\frac{{{J_1}\left( u_0 \right)}}{{2u_0}} - 36\frac{{{J_3}\left( u_0 \right)}}{{u_{0}^2}}} \right)^2},
\end{equation}
where $ {G} $ denotes the maximum satellite antenna gain,
$ {u_0} = 2.07123\frac{{\sin \left( {{\alpha}} \right)}}{{\sin \left( {\alpha _{3{\rm{dB}}}^{}} \right)}} $
with  $ {{\alpha}} $ being the elevation angle between the beam center and ${\rm{S}}{{\rm{U}}}$ and ${\alpha _{3{\rm{dB}}}^{}}$ being the 3 dB angle of satellite beam. Additionally, $ {J_1}\left(  \cdot  \right) $ and $ {J_3}\left(  \cdot  \right) $ are the first-kind Bessel functions of order 1 and 3, respectively. $ \beta $ is modeled as a log-normal random variable, i.e., $ \ln \left( {{\beta _{dB}}} \right) \sim \mathcal{N}\left( {u,{\delta ^2}} \right) $ with $ {{\beta _{dB}}} $ being the dB form of $ \beta $. 
Particularly, $ {{\bf{h}}_{su}} \in \mathbb{C}^{N \times 1}$, ${{\bf{h}}_{tu}}\in \mathbb{C}^{N \times 1}$, ${{\bf{h}}_e}\in \mathbb{C}^{N \times 1}$ are assumed to be the channel vectors from satellite to SU, TU, and Eve, respectively.

Whereas, we adopt the channel model for terrestrial links as $ {\bf{g}} = \sqrt \alpha  {{\bf{g}}_0}$, where  $ \alpha$ denotes the large-scale fading, $ \alpha  = {C_0}{r^{ - 4}} $ with $ {C_0} $ being the channel power gain at the reference
distance of 1 m and $ r $ denoting the distance from BS to the destination, and $ {{\bf{g}}_0} $ denotes the small-scale fading which undergoes Nakagami-$ m $ fading with fading severity $ m $ and average power $ \Omega $.
Particularly, $ {\bf{g}}_{su}\in \mathbb{C}^{M \times 1}$, $ {\bf{g}}_{gu}\in \mathbb{C}^{M \times 1}$, and $ {\bf{g}}_{e}\in \mathbb{C}^{M \times 1}$ denote the channel vectors between BS and SU, TU, and Eve, respectively. 

In addition, the imperfect CSI of Eve is modeled as $ {{\bf{h}}_e} = {{{\bf{\hat h}}}_e} + \Delta {{\bf{h}}_e}$ and $ {{\bf{g}}_e} = {{{\bf{\hat g}}}_e} + \Delta {{\bf{g}}_e} $, where ${\bf{\hat h}}_e$ and ${{{\bf{\hat g}}}_e}$ denote the estimations of eavesdropping channels by satellite and the BS, with $ {\Delta {{\bf{h}}_e}} $ and $ \Delta {{\bf{g}}_e} $ being the norm-bounded estimate errors correspondingly. To facilitate the easy, we assume $ \left\| {\Delta {{\bf{h}}_e}} \right\| = \left\| {\Delta {{\bf{g}}_e}} \right\| \le \Delta  $ in this work.

	\setcounter{equation}{12}
\begin{figure*}[b]
	\hrulefill
	\begin{equation}
		R_s^{su,k} = {\log _2}\left( {\frac{{\sum\limits_i^N {{\rm{Tr}}\left( {{{\bf{H}}_{su,i}}{{\bf{W}}_i}} \right)}  + {\rm{Tr}}\left( {{{\bf{G}}_{su,k}}{{\bf{F}}_k}} \right){\rm{ + }}1}}{{\sum\limits_{i \ne k}^N {{\rm{Tr}}\left( {{{\bf{H}}_{su,i}}{{\bf{W}}_i}} \right)}  + {\rm{Tr}}\left( {{{\bf{G}}_{su,k}}{{\bf{F}}_k}} \right){\rm{ + }}1}}} \right) 
		- {\log _2}\left( {\frac{{\sum\limits_i^N {{\rm{Tr}}\left( {{{\bf{H}}_{e,i}}{{\bf{W}}_i}} \right)}  + {\rm{Tr}}\left( {{{\bf{G}}_{e,k}}{{\bf{F}}_k}} \right){\rm{ + }}1}}{{\sum\limits_{i \ne k}^N {{\rm{Tr}}\left( {{{\bf{H}}_{e,i}}{{\bf{W}}_i}} \right)}  + {\rm{Tr}}\left( {{{\bf{G}}_{e,k}}{{\bf{F}}_k}} \right){\rm{ + }}1}}} \right).\label{Eq-10}
	\end{equation}
	\begin{equation}
		R_s^{tu,k} = {\log _2}\left( {\frac{{{\rm{Tr}}\left( {{{\bf{G}}_{tu,k}}{{\bf{F}}_k}} \right) + \sum\limits_{i = 1}^N {{\rm{Tr}}\left( {{{\bf{H}}_{tu,i}}{{\bf{W}}_i}} \right)}  + 1}}{{\sum\limits_{i = 1}^N {{\rm{Tr}}\left( {{{\bf{H}}_{tu,i}}{{\bf{W}}_i}} \right)}  + 1}}} \right) 
		- {\log _2}\left( {\frac{{{\rm{Tr}}\left( {{{\bf{G}}_{e,k}}{{\bf{F}}_k}} \right) + \sum\limits_{i = 1}^N {{\rm{Tr}}\left( {{{\bf{H}}_{e,i}}{{\bf{W}}_i}} \right)} {\rm{ + }}1}}{{\sum\limits_{i = 1}^N {{\rm{Tr}}\left( {{{\bf{H}}_{e,i}}{{\bf{W}}_i}} \right)} {\rm{ + }}1}}} \right).
		\label{Rtuk}
	\end{equation}
	\begin{equation}
		\sum\limits_{k = 1}^N {R_s^{su,k}}  = {\rm{lo}}{{\rm{g}}_2}(\mathop \Pi \limits_{k = 1}^N \frac{{\sum\limits_{i = 1}^N {{\rm{Tr}}\left( {{{\bf{H}}_{su,i}}{{\bf{W}}_i}} \right)}  + {\rm{Tr}}\left( {{{\bf{G}}_{su,k}}{{\bf{F}}_k}} \right){\rm{ + }}1}}{{\sum\limits_{i \ne k}^N {{\rm{Tr}}\left( {{{\bf{H}}_{su,i}}{{\bf{W}}_i}} \right)}  + {\rm{Tr}}\left( {{{\bf{G}}_{su,k}}{{\bf{F}}_k}} \right){\rm{ + }}1}}) 
		- {\rm{lo}}{{\rm{g}}_2}(\mathop \Pi \limits_{k = 1}^N \frac{{\sum\limits_{i = 1}^N {{\rm{Tr}}\left( {{{\bf{H}}_{e,i}}{{\bf{W}}_i}} \right)}  + {\rm{Tr}}\left( {{{\bf{G}}_{e,k}}{{\bf{F}}_k}} \right){\rm{ + }}1}}{{\sum\limits_{i \ne k}^N {{\rm{Tr}}\left( {{{\bf{H}}_{e,i}}{{\bf{W}}_i}} \right)}  + {\rm{Tr}}\left( {{{\bf{G}}_{e,k}}{{\bf{F}}_k}} \right){\rm{ + }}1}}).\label{sumsr}
	\end{equation}
\end{figure*}



Assuming that $ x_{su,i} $ and $ x_{tu,k} $ denote the expected signal of SU and TU respectively. Without loss of generality, the signal received by $ \text{SU}_k $ in the $k^{th}$ beam can be represented as
	\setcounter{equation}{3}
\begin{equation}\label{sig_su}
{y_{su,k}} = {\bf{h}}_{su,k}^\dag \sum\limits_i^N {{{\bf{w}}_i}{x_{su,i}}}  + {\bf{g}}_{su,k}^\dag {{\bf{f}}_k}{x_{tu,k}} + {n_{su,k}},
\end{equation}
where $ {\bf{h}}_{su,k} $ denotes the channel vector from satellite to $ \text{SU}_k $, $ {\bf{g}}_{su,k} $ denotes the channel vector from the BS to $ \text{TU}_k $, $ {{\bf{w}}_k} \in \mathbb{C}^{N\times 1}$ and $ {{\bf{f}}_k} \in \mathbb{C}^{M\times 1}$ denote the BF vectors of satellite and BS in the $k^{th}$ beam, $ {n_{su,k}} $ is the noise received by $ \text{SU}_k $. 
The signal received by $ \text{TU}_k $ can be written as 
\begin{equation}\label{sig_gu}
{y_{tu,k}} = {\bf{g}}_{tu,k}^\dag {{\bf{f}}_k}{x_{tu,k}} + {\bf{h}}_{tu,k}^\dag \sum\limits_i^N {{{\bf{w}}_i}{x_{su,i}}}  + {n_{tu,k}},
\end{equation}
where $ {{\bf{f}}_k} \in \mathbb{C}^{M\times 1}$ denotes the BF at BS.

Besides, the received signal by the Eve in the $k^{th}$ beam is given by 
	\setcounter{equation}{5}
\begin{equation}\label{sig_se}
{y_{e,k}} = {\bf{h}}_{e,k}^\dag \sum\limits_{i = 1}^N {{{\bf{w}}_i}{x_{su,i}}}  + {\bf{g}}_{e,k}^\dag {{\bf{f}}_k}{x_{tu,k}} + {n_{e,k}}.
\end{equation}

\subsection{Problem Formulation}
From (\ref{sig_se}), it can be see that the Eve wiretaps signals of SU and TU.
However, (\ref{sig_su}) and (\ref{sig_gu}) indicate that the SU and TU receive interference from BS and satellite respectively, which can be designed to unequally degrade the legitimate users and Eve in each beam. Since such interference can be regarded as the green interference to assist the implementation of physical layer security, the secrecy rate of SU and TU can be obtained as follows. 

Based on (\ref{sig_su}) and (\ref{sig_gu}), the SINR of $ \text{SU}_k $ can be calculated as
\begin{equation}\label{SINR_satellite}
{\gamma _{su,k}}{\rm{ = }}\frac{{{{\left\| {{\bf{h}}_{su,k}^\dag {{\bf{w}}_k}} \right\|}^{\rm{2}}}}}{{\sum\limits_{i \ne k}^N {{{\left\| {{\bf{h}}_{su,i}^\dag {{\bf{w}}_i}} \right\|}^{\rm{2}}}}  + {{\left\| {{\bf{g}}_{su,k}^\dag {{\bf{f}}_k}} \right\|}^{\rm{2}}}{\rm{ + }}\delta _{su,k}^2}},
\end{equation}
and the SINR of $ \text{TU}_k $ is obtained as
\begin{equation}\label{SINR_gu}
{\gamma _{tu,k}}{\rm{ = }}\frac{{{{\left\| {{\bf{g}}_{tu,k}^\dag {{\bf{f}}_k}} \right\|}^{\rm{2}}}}}{{\sum\limits_{i = 1}^N {{{\left\| {{\bf{h}}_{tu,i}^\dag {{\bf{w}}_i}} \right\|}^{\rm{2}}}}  + \delta _{tu,k}^2}},
\end{equation}
where $\delta _{su,k}^2$ and $\delta _{tu,k}^2  $ denote the noise power received by $ \text{SU}_k $ and $ \text{TU}_k $.

Accordingly, the SINR of Eve for wiretapping $ \text{SU}_k $ can be written as
\begin{equation}\label{SINR_Eve1}
{\gamma _{se,k}}{\rm{ = }}\frac{{{{\left\| {{\bf{h}}_{e,k}^\dag {{\bf{w}}_k}} \right\|}^{\rm{2}}}}}{{\sum\limits_{i \ne k}^N {{{\left\| {{\bf{h}}_{e,i}^\dag {{\bf{w}}_i}} \right\|}^{\rm{2}}}}  + {{\left\| {{\bf{g}}_{e,k}^\dag {{\bf{f}}_k}} \right\|}^{\rm{2}}}{\rm{ + }}\delta _{e,k}^2}},
\end{equation}
and the SINR of Eve for wiretapping $ \text{TU}_k $ is given by
\begin{equation}\label{SINR_Eve2}
{\gamma _{te,k}}{\rm{ = }}\frac{{{{\left\| {{\bf{g}}_{e,k}^\dag {{\bf{f}}_k}} \right\|}^{\rm{2}}}}}{{\sum\limits_{i = 1}^N {{{\left\| {{\bf{h}}_{e,i}^\dag {{\bf{w}}_i}} \right\|}^{\rm{2}}}} {\rm{ + }}\delta _{e,k}^2}},
\end{equation}
where $\delta _{e,k}^2$ denotes the noise power received by the Eve.

The secrecy rate of $ \text{SU}_k $ and $ \text{TU}_k $ can be respectively given by 
\begin{equation}\label{SR_su}
R_s^{su,k} = {\left[ {{{\log }_2}\left( {1 + {\gamma _{su,k}}} \right) - {{\log }_2}\left( {1 + {\gamma _{se,k}}} \right)} \right]^ + },
\end{equation}
\begin{equation}\label{SR_gu}
R_s^{tu,k} = {\left[ {{{\log }_2}\left( {1 + {\gamma _{tu,k}}} \right) - {{\log }_2}\left( {1 + {\gamma _{te,k}}} \right)} \right]^ + }.
\end{equation}

Based on (\ref{SINR_satellite}--\ref{SINR_Eve2}), the secrecy rates in (\ref{SR_su}) and  (\ref{SR_gu}) can be further represented in (\ref{Eq-10}) and (\ref{Rtuk}) as shown at the bottom of this page, where $ {{\bf{H}}_{su,i}} = {\bf{h}}_{su,k}^{}{\bf{h}}_{su,k}^\dag  $, $ {{\bf{H}}_{tu,i}} = {\bf{h}}_{tu,k}^{}{\bf{h}}_{tu,k}^\dag $, $  {{\bf{H}}_{e,i}} = {\bf{\hat h}}_{e,i}^{}{\bf{\hat h}}_{e,i}^\dag   $, $ {{\bf{W}}_i} = {{\bf{w}}_i}{\bf{w}}_i^\dag  $, $ {{\bf{G}}_{e,k}} = {\bf{\hat g}}_{e,k}^{}{\bf{\hat g}}_{e,k}^\dag $, $ {{\bf{G}}_{tu,k}} = {\bf{g}}_{tu,k}^{}{\bf{g}}_{tu,k}^\dag  $, $ {{\bf{G}}_{su,k}} = {\bf{g}}_{su,k}^{}{\bf{g}}_{su,k}^\dag  $, $ {{\bf{F}}_k} = {{\bf{f}}_k}{\bf{f}}_k^\dag  $.

From (\ref{SINR_Eve1}) and (\ref{SINR_Eve2}), we can see the SINR of Eve in each beam which targets SU is degraded by the green interference from BS, and meanwhile the SINR of Eve targeting TU is degraded by the green interference from satellite. Further, by using (\ref{Eq-10}), the sum secrecy rate of legitimate SUs in satellite multi-beam is given in (\ref{sumsr}).

To guarantee the secrecy performance of both satellite link and terrestrial link simultaneously, our focus is on optimizing the BFs of multi-beam satellite and BS to maximize the sum secrecy rate of legitimate SUs in satellite multi-beam and satisfy a predefined secrecy constraint of legitimate TU in each beam.

Thus, the problem can be mathematically formulated as 
\begin{align}
{\mathcal{P}1:}\quad& \mathop {\max}\limits_{{\left\{ {{{\bf{w}}_k},{{\bf{f}}_k}} \right\}_{k = 1}^N}} \quad \sum\limits_{k = 1}^N {R_s^{su,k}}  \tag{16a} \label{p1-1}\\
\text{s.t.}\quad
&R_s^{tu,k} \ge {Q}_{tu,k},k= 1,2,...,N \tag{16b}, \label{p1-2}\\
& \sum\limits_{k = 1}^N {{{\left\| {{{\bf{w}}_k}} \right\|}^{\rm{2}}}}  \le {P_S} \tag{16c}, \label{p1-3}\\
& {\left\| {{{\bf{f}}_k}} \right\|^{\rm{2}}} \le {P_B},k= 1,2,...,N \tag{16d}. \label{p1-4}
\end{align}

In $ {\mathcal{P}1} $, (\ref{p1-2}) guarantees the secrecy requirements of TUs within satellite beams, where $ {Q}_{tu,k} $ is a predefined secrecy rate threshold for TU in the $ k^{th} $ beam;
(\ref{p1-3}) and (\ref{p1-4}) represent the power constraints of satellite and BS, respectively.

Obviously, the problem $ {\mathcal{P}1} $ has non-convex objective function and constraints. A series of reformulations are conducted in the following section to convert such intractable problem into a solvable alternative, and a joint satellite-terrestrial BF optimizing approach is proposed to solve this problem. To facilitate the simplification, the noise power is normalized, i.e., $ \delta _{su,k}^2 = \delta _{tu,k}^2 =\delta _{e,k}^2= 1$.

\section{Joint Satellite-terrestrial BF Optimizing}
In this section, we aim to solve the secrecy rate maximization problem by reformulating the objection function and its non-convex constraints simplify this problem as a solvable form.
Particularly, a Taylor expansion and SDR are first adopted to reformulate the problem $ {\mathcal{P}1 }$ and then a successive convex approximation approach is carried out to solve the reformulated problem. Finally, the tightness of the relaxation is proved. 
\subsection{Taylor Expansion and SDR Reformulation}
 We introduce exponential variables to make the following changes.
\newcounter{esk}                        
\setcounter{esk}{\value{equation}} 
\setcounter{equation}{16}
\begin{align}
{e^{{s_k}}} = \sum\limits_{i = 1}^N {{\rm{Tr}}\left( {{{\bf{H}}_{su,i}}{{\bf{W}}_i}} \right)}  + {\rm{Tr}}\left( {{{\bf{G}}_{su,k}}{{\bf{F}}_k}} \right){\rm{ + }}1,\label{sk}\\
{e^{{\mu _k}}} = \sum\limits_{i \ne k}^N {{\rm{Tr}}\left( {{{\bf{H}}_{su,i}}{{\bf{W}}_i}} \right)}  + {\rm{Tr}}\left( {{{\bf{G}}_{su,k}}{{\bf{F}}_k}} \right){\rm{ + }}1,\label{tk}\\
{e^{{q_k}}} = \sum\limits_{i = 1}^N {{\rm{Tr}}\left( {{{\bf{H}}_{e,i}}{{\bf{W}}_i}} \right)}  + {\rm{Tr}}\left( {{{\bf{G}}_{e,k}}{{\bf{F}}_k}} \right){\rm{ + }}1,\label{qk}\\
{e^{{v_k}}} = \sum\limits_{i \ne k}^N {{\rm{Tr}}\left( {{{\bf{H}}_{e,i}}{{\bf{W}}_i}} \right)}  + {\rm{Tr}}\left( {{{\bf{G}}_{e,k}}{{\bf{F}}_k}} \right){\rm{ + }}1.\label{vk}
\end{align}

By substituting (\ref{sk}--\ref{vk}) into (\ref{p1-1}), the objection function in $ {\mathcal{P}1} $ can be equivalently represented as
	\setcounter{equation}{20}
\begin{equation}\label{Eq-newp}
\mathop {\max }\limits_{\left\{ {{s_k},{\mu _k},{q_k},{v_k}} \right\}_{k = 1}^N} \sum\limits_{k = 1}^N {\left( {{s_k} - {\mu _k} - {q_k} + {v_k}} \right)} , 
\end{equation}
which is a convex problem because the criterion is a sum of affine functions (composed with $ {{s_k} - {\mu _k}} $ and $ {v_k} - {q_k} $). Particularly, the constrains of $ {s_k}, \mu_k, q_k, v_k$ in (\ref{Eq-newp}) hold the following bounds. 
\begin{align}
{e^{{s_k}}} \le \sum\limits_{i = 1}^N {{\rm{Tr}}\left( {{{\bf{H}}_{su,i}}{{\bf{W}}_i}} \right)}  + {\rm{Tr}}\left( {{{\bf{G}}_{su,k}}{{\bf{F}}_k}} \right){\rm{ + }}1,\label{Eq-16}\\
{e^{{\mu _k}}} \ge \sum\limits_{i \ne k}^N {{\rm{Tr}}\left( {{{\bf{H}}_{su,i}}{{\bf{W}}_i}} \right)}  + {\rm{Tr}}\left( {{{\bf{G}}_{su,k}}{{\bf{F}}_k}} \right){\rm{ + }}1,\label{Eq-17}\\
{e^{{q_k}}} \ge \sum\limits_{i = 1}^N {{\rm{Tr}}\left( {{{\bf{H}}_{e,i}}{{\bf{W}}_i}} \right)}  + {\rm{Tr}}\left( {{{\bf{G}}_{e,k}}{{\bf{F}}_k}} \right){\rm{ + }}1,\label{Eq-18}\\
{e^{{v_k}}} \le \sum\limits_{i \ne k}^N {{\rm{Tr}}\left( {{{\bf{H}}_{e,i}}{{\bf{W}}_i}} \right)}  + {\rm{Tr}}\left( {{{\bf{G}}_{e,k}}{{\bf{F}}_k}} \right){\rm{ + }}1.\label{Eq-19}
\end{align}

It can be verified that all the inequalities from (\ref{Eq-16}) to (\ref{Eq-19}) hold with equalities at the optimal points by the monotonicity of objective function. However, it can be observed that the constraints in (\ref{Eq-17}) and (\ref{Eq-18}) are still non-convex.

In addition, keeping (\ref{Rtuk}) in mind, the constraint in (\ref{p1-2}) is also non-convex due to the non-convex fractional programming.
Similarly, we make the following changes.
\begin{equation}
{e^{{\tau _k}}} = \sum\limits_{i=1}^N {{\rm{Tr}}\left( {{{\bf{H}}_{tu,i}}{{\bf{W}}_i}} \right)} {\rm{ + Tr}}\left( {{{\bf{G}}_{tu,k}}{{\bf{F}}_k}} \right) + 1,\label{Eqtao}
\end{equation}
\begin{equation}
{e^{{\eta _k}}}  = \sum\limits_{i = 1}^N {{\rm{Tr}}\left( {{{\bf{H}}_{tu,i}}{{\bf{W}}_i}} \right)}  + 1, \label{Eq23}
\end{equation} 
\begin{equation}
{e^{{\alpha _k}}}  = \sum\limits_{i = 1}^N {{\rm{Tr}}\left( {{{\bf{H}}_{e,i}}{{\bf{W}}_i}} \right)} {\rm{ + }}1. \label{Eqal}
\end{equation} 

By substituting (\ref{Eqtao}--\ref{Eqal}, \ref{qk}) into (\ref{Rtuk}), the secrecy constraint of TU in (\ref{p1-2}) can be equivalently reformulated as
\begin{equation}\label{Eq24}
{\eta _k} + {q_k} - {\tau _k} - {\alpha _k} \le  - \frac{{{Q_{tu,k}}}}{{{{\log }_2}e}},
\end{equation}
with the successive constraints as follows  
\begin{align}
{e^{{\tau _k}}} &\le \sum\limits_i^N {{\rm{Tr}}\left( {{{\bf{H}}_{tu,i}}{{\bf{W}}_i}} \right)} {\rm{ + Tr}}\left( {{{\bf{G}}_{tu,k}}{{\bf{F}}_k}} \right) + 1, \label{Eq26-6}\\
{e^{{\eta _k}}} &\ge \sum\limits_{i = 1}^N {{\rm{  Tr}}\left( {{{\bf{H}}_{tu,i}}{{\bf{W}}_i}} \right)}  + 1,\label{Eq26-7}\\
{e^{{\alpha _k}}}& \le \sum\limits_{i = 1}^N {{\rm{Tr}}\left( {{{\bf{H}}_{e,i}}{{\bf{W}}_i}} \right)} {\rm{ + }}1,\label{Eq26-8}
\end{align} 
and (\ref{Eq-18}) is also satisfied, where (\ref{Eq26-7}) and (\ref{Eq-18}) are non-convex.

To address the non-convex constraints in (\ref{Eq-17}), and (\ref{Eq-18}), (\ref{Eq26-7}), we adopt Taylor expansion to make these constraints conservatively convex approximating at $\left( {\left\{ {{{\tilde \mu }_k}} \right\},\left\{ {{{\tilde q}_k}} \right\},\left\{ {{{\tilde \eta }_k}} \right\}} \right)$. By using the first-order Taylor expansion of $ {e^{{\mu _k}}},{e^{{q_k}}}$, and ${e^{{\eta _k}}}$, the restrictive approximations for (\ref{Eq-17}), and (\ref{Eq-18}), (\ref{Eq26-7}) can be given as 
\begin{equation}
	\sum\limits_{i \ne k}^N {{\rm{Tr}}\left( {{{\bf{H}}_{su,i}}{{\bf{W}}_i}} \right)}  + {\rm{Tr}}\left( {{{\bf{G}}_{su,k}}{{\bf{F}}_k}} \right){\rm{ + }}1
	\le {e^{{{\tilde \mu }_k}}}\left( {{\mu _k} - {{\tilde \mu }_k} + 1} \right), \label{Eq-27}
\end{equation}
\setcounter{equation}{39}
\begin{figure*}[b]
	\hrulefill
	\begin{equation}\label{SINR_Eve4}
		{\gamma _{se,k}}{\rm{ = }}\frac{{{{\left\| {{\bf{h}}_{e,k}^\dag {{\bf{w}}_k}} \right\|}^{\rm{2}}}}}{{\sum\limits_{i \ne k}^N {{{\left\| {{\bf{h}}_{e,i}^\dag {{\bf{w}}_i}} \right\|}^{\rm{2}}}}  + {\ell _k} {P_B}{{\left\| {{\bf{g}}_{e,k}^\dag {\bf{f}}_k^{mrt}} \right\|}^{\rm{2}}}{\rm{ + }}\left( {1 - {\ell _k} } \right){P_B}{{\left\| {{\bf{g}}_{e,k}^\dag {{\bm{{\rm O}}}_k}{{\bf{v}}_k}} \right\|}^2}{\rm{ + }}\delta _{su,k}^2}}, 
	\end{equation}
\end{figure*}
\setcounter{equation}{33}
\begin{equation}
	\sum\limits_{i = 1}^N {{\rm{Tr}}\left( {{{\bf{H}}_{e,i}}{{\bf{W}}_i}} \right)}  + {\rm{Tr}}\left( {{{\bf{G}}_{e,k}}{{\bf{F}}_k}} \right){\rm{ + }}1\le {e^{{{\tilde q}_k}}}\left( {{q_k} - {{\tilde q}_k} + 1} \right),\label{Eq-28}
\end{equation}
\begin{equation}
	\sum\limits_{i=1}^N {{\rm{Tr}}\left( {{{\bf{H}}_{tu,i}}{{\bf{W}}_i}} \right)}  + 1{\rm{ }} \le {e^{{{\tilde \eta }_k}}}\left( {{\eta _k} - {{\tilde \eta }_k} + 1} \right), \label{Eq-32}
\end{equation}
where $ {{\tilde \mu}_k} $, $ {{\tilde q}_k} $, and $ {{{\tilde \eta }_k}} $ begin with an initial values.  
\subsection{Successive Convex Approximation Approach}
 By defining the renewed optimization variables as a set, i.e., 
\begin{equation}\label{adx1}
	 \Im =\left\{ {\left\{ {{{\bf{W}}_k}} \right\},\left\{ {{{\bf{F}}_k}} \right\},\left\{ {{s_k}} \right\},\left\{ {{\mu _k}} \right\},\left\{ {{q_k}} \right\},\left\{ {{v_k}} \right\},\left\{ {{\tau _k}} \right\},\left\{ {{\eta _k}} \right\}} \right\}, \nonumber
\end{equation}
thus the primal problem $ {\mathcal{P}1} $ can be reformulated as 
 \begin{align}
 {\mathcal{P}2:}\quad& \mathop {\max}\limits_{{ \Im }} \quad \sum\limits_{k = 1}^N {\left( {{s_k} - {\mu_k} - {q_k} + {v_k}} \right)}  \tag{36a} \label{p22-1}\\
\text{s.t.}\quad
& \sum\limits_{k = 1}^N {{\rm{Tr}}\left( {{{\bf{W}}_k}} \right)}  \le {P_S} ,k= 1,2,...,N \tag{36b}, \label{p2-222}\\
& {{\rm{Tr}}\left( {{{\bf{F}}_k} } \right)}  \le {P_B}, k= 1,2,...,N, \tag{36c}\label{p2-4}\\
& {{\bf{F}}_k} \succeq 0,{\bf{W}}_k \succeq 0, \tag{36d} \label{p2-6}\\
& (\ref{Eq-16}), (\ref{Eq-19}), (\ref{Eq24}-\ref{Eq26-6}), (\ref{Eq26-8}), (\ref{Eq-27}-\ref{Eq-32}).  \tag{36e} \label{p2-5}
\end{align}

To solve ${\mathcal{P}2}$, we adopt the cvx tool and carry out a SCA based approach, where  $ {{\tilde \mu}_k} $, $ {{\tilde q}_k} $, and $ {{{\tilde \eta }_k}} $ can be updated by each iteration of SCA. Particularly, the details of SCA-based joint satellite-terrestrial BF optimization can be seen in algorithm table 1. In addition, the main computational complexity is solving the convex approximation problem in each iteration. Considering the
 SDP for solving 
$ {{\bf{w}}_n^\star}$ and $ {{\bf{f}}_n^\star}$ can be calculated by $ {\cal{O}} (max\{m,n\}^\text{4}n^\text{1/2}) $, where $m$ and $n$ are the constraint order and the dimension of equality constraints for SDP, respectively. Thus, the total complexity can be calculated as $ t\cdot ({\cal{O}} ((8N+2)^\text{4}) +log(1/\epsilon))$.

\begin{algorithm}
\caption{SCA-based joint satellite-terrestrial BF optimization}	
\KwIn{$\{ {Q}_{tu,k} \}, \varepsilon, P_S, P_B$.}
\KwResult{\{$ {{\bf{w}}_n^\star}, n =1, ..., N \}$, $ \{{{\bf{f}}_k^\star}, k =1, ..., N $\}.}
\textbf{Initialization}: $\left\{ {\tilde \eta _k^0 } \right\},\left\{ {\tilde \mu _k^0 } \right\},\left\{ {\tilde q_k^0 } \right\}$.\\
Set step $ t=0 $;\\
\Repeat{$\left| {R_{s,sum}^t - R_{s,sum}^{t - 1}} \right| < \epsilon $}
{

	
		Using the CVX solver SDP to solve $ \mathcal{P}2 $; \\
		\KwOut{$ \left\{ {s_k^ \star } \right\}, \left\{ { \eta _k^ \star } \right\}, \left\{ { \mu _k^ \star } \right\}, \left\{ { q_k^ \star } \right\},\left\{ {v_k^ \star } \right\} $, $\left\{ {{\bf{W}}_k^ \star } \right\},\left\{ {{\bf{F}}_k^ \star } \right\}$;}
		
		Obtain $ R_{s,sum}^t = \sum\limits_{k = 1}^N {s_k^ \star  - } \mu _k^ \star  - q_k^ \star  + v_k^ \star $ .\\
		$ \left\{ {\tilde \eta _k^t = \eta _k^ \star ;\tilde \mu _k^t = \mu _k^ \star ;\tilde q_k^t = q_k^ \star } \right\} $.\\
}
Obtain $ \left\{ {{\bf{w}}_n^ \star } \right\}$ and $\left\{ {{\bf{f}}_n^ \star } \right\}$ by the singular value decomposition (SVD) of $ \left\{ {{\bf{W}}_k^ \star } \right\}$ and $\left\{ {{\bf{F}}_k^ \star } \right\}$.\\
\textbf{Procedure End}	
	
\end{algorithm}

In addition, to analyze the tightness of SDR from ${\mathcal{P}1}$ to ${\mathcal{P}2}$, the rank-one of $ \{{{\bf{W}}_k}$\} and $\{{{\bf{F}}_k}\} $ is proved as follows.
\begin{proposition}
We consider a power minimization problem, which can be expressed as 
\begin{align}
	{\mathcal{P}3:}\quad& \mathop {\min }\limits_\Im  \sum\limits_{i = 1}^N {{\rm{Tr}}\left( {{{\bf{W}}_i}} + {{{\bf{F}}_i}}\right)}  \tag{37a} \label{p3-1}\\
	\text{s.t.} \quad & \sum\limits_{k = 1}^N {\left( {{s_k} - {\mu _k} - {q_k} + {v_k}} \right)}  \ge {\varphi ^\circ } \tag{37b}, \label{p3-2}\\
	& (\ref{p2-222}-\ref{p2-5}), \tag{37c} \label{p3-3}
\end{align}
where $ \varphi ^\circ  $ denotes the optimal objective value of ${\mathcal{P}2}$.
It can be obtained that any feasible solutions of ${\mathcal{P}3}$ is optimal for ${\mathcal{P}2}$.
\begin{proof}
	Assuming that ${\Im^\star } $ is the optimal solution for solving ${\mathcal{P}3}$, then the following condition is restrictively satisfied
	\setcounter{equation}{37}
	\begin{equation}\label{EqPro1}
		{\varphi^\circ} \le \sum\limits_{k = 1}^N {\left( {s_k^\star  - {\mu _k^\star} - {q_k^\star} + {v_k^\star}} \right)},  
	\end{equation}
with $ {\bf{W}}_k^ \star  $ and $ {\bf{F}}_k^ \star $ satisfying constraints in (\ref{p2-222}-\ref{p2-5}).

From (\ref{EqPro1}), it can be found that the maximum $ {\varphi^\circ} $ is reached when $ \sum\limits_{k = 1}^N {\left( {{s_k} - {\mu_k} - {q_k} + {v_k}} \right)} $ has the optimal solution at ${\Im^\star } $. Thus, ${\Im^\star } $ is the optimal solution for ${\mathcal{P}2}$ with objective value $ {\varphi^\circ} $ and the proof is concluded.
\end{proof}
\end{proposition} 
Based on Proposition 1, we indirectly prove the rank-one condition in ${\mathcal{P}2}$ by proving that in ${\mathcal{P}3}$ by the following theorem.
\begin{theorem}
	For any feasible solutions of $ {{{\bf{W}}_k}} $ and $ {{{\bf{F}}_k}} $ from ${\mathcal{P}3}$, $ Rank\left( {{{\bf{W}}_k}} \right) = 1 $ and $ Rank\left( {{{\bf{F}}_k}} \right) = 1 $.
\end{theorem}
\begin{proof}
	Please see the appendix A.
\end{proof}

\section{Joint Optimization of Satellite BF and BS PA}
For the reference propose, an alternative approach that jointly optimizing satellite BF and the PA of BS is given in this section, which is as a  benchmark in this paper.

We consider the BSs can use a partial transmission power to generate AN for confusing Eve deliberately. The SINRs of SU, TU, and Eve can be respectively rewritten as 
\begin{equation}\label{SU2}
	{\gamma _{su,k}}{\rm{ = }}\frac{{{{\left\| {{\bf{h}}_{su,k}^H{{\bf{w}}_k}} \right\|}^{\rm{2}}}}}{{\sum\limits_{i \ne k}^N {{{\left\| {{\bf{h}}_{su,i}^H{{\bf{w}}_i}} \right\|}^{\rm{2}}}}  + {\ell _k} {P_B}{{\left\| {{\bf{g}}_{su,k}^H{\bf{f}}_k^{mrt}} \right\|}^{\rm{2}}}{\rm{ + }}\delta _{su,k}^2}},
\end{equation}
	\setcounter{equation}{40}
\begin{equation}\label{TU2}
{\gamma _{tu,k}}{\rm{ = }}{\ell _k}{P_B}{\left\| {{\bf{g}}_{tu,k}^\dag {\bf{f}}_k^{mrt}} \right\|^{\rm{2}}}/\big( {\sum\limits_{i = 1}^N {{{\left\| {{\bf{h}}_{tu,i}^\dag {{\bf{w}}_i}} \right\|}^{\rm{2}}}}  + \delta _{tu,k}^2} \big),
\end{equation}
\begin{equation}\label{SINR_Eve3}
{\gamma _{te,k}}{\rm{ = }}\frac{{{\ell _k} {P_B}{{\left\| {{\bf{g}}_{e,k}^\dag {\bf{f}}_k^{mrt}} \right\|}^{\rm{2}}}}}{{\sum\limits_{i = 1}^N {{{\left\| {{\bf{h}}_{e,i}^\dag {{\bf{w}}_i}} \right\|}^{\rm{2}}}} {\rm{ + }}\left( {1 - {\ell _k} } \right){P_B}{{\left\| {{\bf{g}}_{e,k}^\dag {{\bm{{\rm O}}}_k}{{\bf{v}}_k}} \right\|}^2}{\rm{ + }}\delta _{tu,k}^2}},	
\end{equation}
where $ {\ell _k} $ is the PA coefficient of BS, $ {\bf{f}}_k^{mrt} = {{{{\bf{g}}_{tu,k}}} {\left/
		{\vphantom {{{{\bf{g}}_{tu,k}}} {\left\| {{{\bf{g}}_{tu,k}}} \right\|}}} \right.
		} {\left\| {{{\bf{g}}_{tu,k}}} \right\|}} $ denotes the MRT-based BF vector of BS for transmitting useful signal to TU, 
	$ {{\bf{v}}_k} \in \mathbb{C}^{(M-1) \times 1}$ is the AN vector.
	and $ {{\bf{O}}_k} = \left( {{\bf{I}} - {\bm{\hbar}}\left( {{{\bm{\hbar}}^\dag }{\bm{\hbar}}} \right){{\bm{\hbar}}^\dag }} \right){{\bf{g}}_{e,k}}$
with $ \bm{\hbar } =
	[{{{\bf{g}}_{tu}}}, {{{\bf{g}}_{su}}}]
$ is the projection matrix into the null of two legitimate channels, i.e., $ {{{\bf{g}}_{tu}}}$ and $  {{{\bf{g}}_{su}}} $. In addition, we make 
	$  {\bf{\Lambda} _k} = {{\bf{O}}_k}{{\bf{v}}_k}{\bf{v}}_k^\dag {\bf{O}}_k^\dag, $
and 
\begin{equation}\label{ask2}
	 {{\bf{A}}_k} = \left( {1 - {\ell _k} } \right){P_B}{\rm{Tr}}\left( {{{\bf{G}}_{e,k}} \bf{{ \Lambda  }_k}} \right). 
\end{equation}
	
Based on (\ref{SU2}--\ref{ask2}), 
the secrecy rate of SU and TU are obtained in(\ref{Rs11}) and (\ref{Rtuk2}) as shown at the top of this page.
\begin{figure*}[t]
	\begin{equation}\label{Rs11}
		R_s^{su,k} = {\log _2}\left( {\frac{{\sum\limits_i^N {{\rm{Tr}}\left( {{{\bf{H}}_{su,i}}{{\bf{W}}_i}} \right)}  + {\ell _k} {P_B}{\rm{Tr}}\left( {{{\bf{G}}_{su,k}}{\bf{F}}_k^{mrt}} \right){\rm{ + }}1}}{{\sum\limits_{i \ne k}^N {{\rm{Tr}}\left( {{{\bf{H}}_{su,i}}{{\bf{W}}_i}} \right)}  + {\ell _k} {P_B}{\rm{Tr}}\left( {{{\bf{G}}_{su,k}}{\bf{F}}_k^{mrt}} \right){\rm{ + }}1}}} \right) 
		- {\log _2}\left( {\frac{{\sum\limits_i^N {{\rm{Tr}}\left( {{{\bf{H}}_{e,i}}{{\bf{W}}_i}} \right)}  + {\ell _k} {P_B}{\rm{Tr}}\left( {{{\bf{G}}_{e,k}}{\bf{F}}_k^{mrt}} \right){\rm{ + }}{{\bf{A}}_k}{\rm{ + }}1}}{{\sum\limits_{i \ne k}^N {{\rm{Tr}}\left( {{{\bf{H}}_{e,i}}{{\bf{W}}_i}} \right)}  + {\ell _k} {P_B}{\rm{Tr}}\left( {{{\bf{G}}_{e,k}}{\bf{F}}_k^{mrt}} \right){\rm{ + }}{{\bf{A}}_k}{\rm{ + }}1}}} \right).
	\end{equation}
	\begin{equation}
		R_s^{tu,k} = {\log _2}\left( {\frac{{{\ell _k} {P_B}{\rm{Tr}}\left( {{{\bf{G}}_{tu,k}}{\bf{F}}_k^{mrt}} \right) + \sum\limits_{i = 1}^N {{\rm{Tr}}\left( {{{\bf{H}}_{tu,i}}{{\bf{W}}_i}} \right)}  + 1}}{{\sum\limits_{i = 1}^N {{\rm{Tr}}\left( {{{\bf{H}}_{tu,i}}{{\bf{W}}_i}} \right)}  + 1}}} \right) 
		- {\log _2}\left( {\frac{{{\ell _k} {P_B}{\rm{Tr}}\left( {{{\bf{G}}_{e,k}}{\bf{F}}_k^{mrt}} \right) + \sum\limits_{i = 1}^N {{\rm{Tr}}\left( {{{\bf{H}}_{e,i}}{{\bf{W}}_i}} \right)} {\rm{ + }}{{\bf{A}}_k}{\rm{ + }}1}}{{\sum\limits_{i = 1}^N {{\rm{Tr}}\left( {{{\bf{H}}_{e,i}}{{\bf{W}}_i}} \right)} {\rm{ + }}{{\bf{A}}_k}{\rm{ + }}1}}} \right).
		\label{Rtuk2}
	\end{equation}
	\hrulefill
\end{figure*} 
Hence, another optimization problem targeting the primal objective and constraints in ${\mathcal{P}1}  $ can be formulated as
\begin{align}
	{\mathcal{P}4:}\quad& \mathop {\max }\limits_{\left\{ {{{\bf{w}}_k}} \right\}_{k = 1}^N, {\ell _k} }  \quad \sum\limits_{k = 1}^N {R_s^{su,k}}  \tag{47a} \label{p4-1}\\
	\text{s.t.}\quad
	&R_s^{tu,k} \ge {Q}_{tu,k},k= 1,2,...,N \tag{47b}, \label{p4-2}\\
	& \sum\limits_{k = 1}^N {{{\left\| {{{\bf{w}}_k}} \right\|}^{\rm{2}}}}  \le {P_S} \tag{47c}, \label{p4-3}\\
	& 0 \le {\ell _k}  \le 1 \tag{47d}. \label{p4-4}
\end{align}

From (\ref{Rs11}--\ref{p4-4}), it is observed that the problem $ {\mathcal{P}4} $ is also intractable since its non-convex objective function and constraints.
Similarly, the Taylor expansion is adopted to reformulate the problem $ {\mathcal{P}4} $ and changes by the corresponding exponential variables are made as follows. 
	\setcounter{equation}{47} 
\begin{equation}
	{e^{{{s'}_k}}} = \sum\limits_i^N {{\rm{Tr}}\left( {{{\bf{H}}_{su,i}}{{\bf{W}}_i}} \right)}  + {\ell _k} {P_B}{\rm{Tr}}\left( {{{\bf{G}}_{su,k}}{\bf{F}}_k^{mrt}} \right){\rm{ + }}1,\label{sk-}
\end{equation}
	\begin{equation}
	{e^{{{\mu '}_k}}} = \sum\limits_{i \ne k}^N {{\rm{Tr}}\left( {{{\bf{H}}_{su,i}}{{\bf{W}}_i}} \right)}  + {\ell _k} {P_B}{\rm{Tr}}\left( {{{\bf{G}}_{su,k}}{\bf{F}}_k^{mrt}} \right){\rm{ + }}1,\label{muk-}
\end{equation}
	\begin{equation}
	{e^{{{q'}_k}}} = \sum\limits_i^N {{\rm{Tr}}\left( {{{\bf{H}}_{e,i}}{{\bf{W}}_i}} \right)}  + {\ell _k} {P_B}{\rm{Tr}}\left( {{{\bf{G}}_{e,k}}{\bf{F}}_k^{mrt}} \right){\rm{ + }}{{\bf{A}}_k}{\rm{ + }}1,\label{qk-}
\end{equation}
	\begin{equation}
	{e^{{{v'}_k}}} = \sum\limits_{i \ne k}^N {{\rm{Tr}}\left( {{{\bf{H}}_{e,i}}{{\bf{W}}_i}} \right)}  + {\ell _k} {P_B}{\rm{Tr}}\left( {{{\bf{G}}_{e,k}}{\bf{F}}_k^{mrt}} \right){\rm{ + }}{{\bf{A}}_k}{\rm{ + }}1,\label{vk-}
\end{equation}
\begin{equation}
{e^{{{\tau '}_k}}} = {\ell _k} {P_B}{\rm{Tr}}\left( {{{\bf{G}}_{tu,k}}{\bf{F}}_k^{mrt}} \right) + \sum\limits_{i = 1}^N {{\rm{Tr}}\left( {{{\bf{H}}_{tu,i}}{{\bf{W}}_i}} \right)}  + 1,\label{tauk-}
\end{equation}
\begin{equation}
{e^{{{\eta '}_k}}} = \sum\limits_{i = 1}^N {{\rm{Tr}}\left( {{{\bf{H}}_{tu,i}}{{\bf{W}}_i}} \right)}  + 1,\label{etak-}
\end{equation}
\begin{equation}
	{e^{{{\alpha '}_k}}} = \sum\limits_{i = 1}^N {{\rm{Tr}}\left( {{{\bf{H}}_{e,i}}{{\bf{W}}_i}} \right)} {\rm{ + }}{{\bf{A}}_k}{\rm{ + }}1.\label{alphak-}
\end{equation}

By defining the renewed optimization variables as a set, i.e., 
\begin{equation}\label{adx2}
	\Im ' = \left\{ {\left\{ {{{\bf{W}}_k}} \right\},{\ell _k} ,\left\{ {{{s'}_k}} \right\},\left\{ {{{\mu '}_k}} \right\},\left\{ {{{q'}_k}} \right\},\left\{ {{{v'}_k}} \right\},\left\{ {{{\tau '}_k}} \right\},\left\{ {{{\eta '}_k}} \right\}} \right\}, \nonumber
\end{equation}
and using the above replacements in (\ref{sk-}--\ref{alphak-}), the problem $ {\mathcal{P}4} $ can be reformulated as  
\begin{align}
	{\mathcal{P}5:}\quad \mathop {\max }\limits_{\Im '} \sum\limits_{k = 1}^N {\left( {{{s'}_k} - {{\mu '}_k} - {{q'}_k} + {{v'}_k}} \right)}   \tag{55a} \label{p50-1}
\end{align}
\begin{align}
	\text{s.t.}\quad
	{{\eta '}_k} + {{q'}_k} - {{\tau '}_k} - {{\alpha '}_k} \le  - \frac{{{Q_{tu,k}}}}{{{{\log }_2}e}}, \tag{55b} \label{p50-2}
\end{align}
\begin{equation}
	 		{e^{{{s'}_k}}} \leq \sum\limits_i^N {{\rm{Tr}}\left( {{{\bf{H}}_{su,i}}{{\bf{W}}_i}} \right)}  + {\ell _k} {P_B}{\rm{Tr}}\left( {{{\bf{G}}_{su,k}}{\bf{F}}_k^{mrt}} \right){\rm{ + }}1,\tag{55c} \label{p50-3}
\end{equation}
\begin{align}
	\sum\limits_{i \ne k}^N {{\rm{Tr}}\left( {{{\bf{H}}_{su,i}}{{\bf{W}}_i}} \right)}  + {\ell _k} {P_B}{\rm{Tr}}\left( {{{\bf{G}}_{su,k}}{\bf{F}}_k^{mrt}} \right){\rm{ + }}1 \nonumber\\
	\le {e^{{{\tilde \mu '}_k}}}\left( {{{\mu '}_k} - {{\tilde \mu '}_k} + 1} \right),\tag{55d}\label{p50-4}
\end{align}
\begin{align}
		\sum\limits_i^N {{\rm{Tr}}\left( {{{\bf{H}}_{e,i}}{{\bf{W}}_i}} \right)}  + {\ell _k} {P_B}{\rm{Tr}}\left( {{{\bf{G}}_{e,k}}{\bf{F}}_k^{mrt}} \right){\rm{ + }}{{\bf{A}}_k}{\rm{ + }}1 \nonumber\\
		\le {e^{{{\tilde q'}_k}}}\left( {{{q'}_k} - {{\tilde q'}_k} + 1} \right),\tag{55e}\label{p50-5}
\end{align}
\begin{equation}
		{e^{{{v'}_k}}} \leq \sum\limits_{i \ne k}^N {{\rm{Tr}}\left( {{{\bf{H}}_{e,i}}{{\bf{W}}_i}} \right)}  + {\ell _k} {P_B}{\rm{Tr}}\left( {{{\bf{G}}_{e,k}}{\bf{F}}_k^{mrt}} \right){\rm{ + }}{{\bf{A}}_k}{\rm{ + }}1,\tag{55f}\label{p50-6}
\end{equation}
\begin{equation}
		{e^{{{\tau '}_k}}} \leq {\ell _k} {P_B}{\rm{Tr}}\left( {{{\bf{G}}_{tu,k}}{\bf{F}}_k^{mrt}} \right) + \sum\limits_{i = 1}^N {{\rm{Tr}}\left( {{{\bf{H}}_{tu,i}}{{\bf{W}}_i}} \right)}  + 1,\tag{55g}\label{p50-7}
	\end{equation}
\begin{equation}
		\sum\limits_{i = 1}^N {{\rm{Tr}}\left( {{{\bf{H}}_{tu,i}}{{\bf{W}}_i}} \right)}  + 1 \le {e^{{{\tilde \eta '}_k}}}\left( {{{\eta '}_k} - {{\tilde \eta '}_k} + 1} \right),\tag{55h}\label{p50-8}
\end{equation}
\begin{equation}		{e^{{{\alpha '}_k}}} \leq \sum\limits_{i = 1}^N {{\rm{Tr}}\left( {{{\bf{H}}_{e,i}}{{\bf{W}}_i}} \right)} {\rm{ + }}{{\bf{A}}_k}{\rm{ + }}1,\tag{55i}\label{p50-9}
\end{equation}
\begin{align} \sum\limits_{k = 1}^N {{\rm{Tr}}\left( {{{\bf{W}}_k}} \right)}  \le {P_S}, \tag{53j}\label{p50-10} \nonumber\\
	 0 \le {\ell _k}  \le 1 \tag{55k}. \label{p50-11}
\end{align}  

From (50), we can see that the problem $ {\mathcal{P}5} $ can also be solved by a SCA-based optimization algorithm. Finally, we conduct the solving procedure in the algorithm table 2. 
\begin{algorithm}
	\caption{SCA-based joint satellite BF and the PA of BS optimization}	
	\KwIn{$\{ {Q}_{tu,k} \}, \varepsilon, P_S, P_B$.}
	\KwResult{\{$ {{\bf{w}}_n^\star}, n =1, ..., N \}$, $ \{{\ell _k}, k =1, ..., N $\}.}
	\textbf{Initialization}: $\left\{ {\tilde \eta '}_k \right\},\left\{ {\tilde \mu' }_k \right\},\left\{ {\tilde q'}_k \right\}$.\\
	Set step $ t=0 $;\\
	Calculate $ {\bf{f}}_k^{mrt} = {{{{\bf{g}}_{tu,k}}} {\left/
			{\vphantom {{{{\bf{g}}_{tu,k}}} {\left\| {{{\bf{g}}_{tu,k}}} \right\|}}} \right.
		} {\left\| {{{\bf{g}}_{tu,k}}} \right\|}} $; \\
	Obtain ${\bf{F}}_k^{mrt} = {\bf{f}}_k^{mrt}{\left( {{\bf{f}}_k^{mrt}} \right)^\dag }$; 
	\\
	\Repeat{$\left| {R_{s,sum}^t - R_{s,sum}^{t - 1}} \right| < \epsilon $}
	{

		
		Using the CVX solver SDP to solve $ \mathcal{P}2 $; \\
		\KwOut{$ \left\{ {{s'}_k^ \star } \right\}, \left\{ { {\eta'} _k^ \star } \right\}, \left\{ { {\mu'} _k^ \star } \right\}, \left\{ { {q'}_k^ \star } \right\},\left\{ {{v'}_k^ \star } \right\} $, $\left\{ {{\bf{W}}_k^ \star } \right\},\left\{ {\ell _k} \right\}$;}
		
		Obtain $ R_{s,sum}^t = \sum\limits_{k = 1}^N {{s'}_k^ \star  - } {\mu'} _k^ \star  - {q'}_k^ \star  + {v'}_k^ \star $ .\\
		Update $ \eta _k $, $\mu _k$, and $ q_k $ by $ \eta _k^ \star $, $\mu _k^ \star$, and $ q_k^ \star   $, respectively.\\
	}
	Obtain $ \left\{ {{\bf{w}}_n^ \star } \right\}$ by the singular value decomposition (SVD) of $ \left\{ {{\bf{W}}_k^ \star } \right\}$.\\
	\textbf{Procedure End}	
	
\end{algorithm}

\section{Performance Evaluation}
\begin{table}[!h]
	\centering
	\caption{System Parameters Setting}
	\label{tab2}
	\begin{tabular}{lll}
		\toprule
		\midrule
		System Parameters  & Numerical Value \\
		\midrule
		\emph{Satellite-to-ground channel parameters}\\
		Satellite height   & 600 Km \\
		Carrier frequency & 2 GHz\\
		Maximum beam gain  & 46.6 dB \\
		3 dB angle (for all beams)  &0.4$^\circ $ \\
		Rain attenuation parameters  &$ {\mu _{{\zeta _{dB}}}} =$ -3.152, ${\delta ^2} =$ 1.6  \\
		\emph{Terrestrial BS channel parameters}\\
		Channel power gain  &  -38.46  dB \\
		Nakagami-$m$ channel parameters&  $ m=2 $, $ \Omega =1 $\\
		\bottomrule
	\end{tabular}
\end{table}
In this section, we conduct extensive numerical simulations to evaluate the secrecy performance.
The main system parameters are set in Table I. Specifically, the height of satellite orbit is 600 Km and the maximum beam gain $ G_{max} $ is set to 46.6 dB. The 3 dB angles of all beams are set to $ {\alpha _{i,3dB}}=0.4^\circ $. The rain attenuation parameters of downlink satellite-terrestrial channels are set to $ {\mu _{{\zeta _{dB}}}} = -3.152$ and ${\delta ^2} = 1.6$. The carrier frequency of satellite downlink transmission is 2 GHz.
The channel power gain of BS downlink at the reference distance of 1 m is $ -38.46 $ dB. The horizontal distance from BS to SU, TU, and Eve is 100 m, 100 m, 120m, respectively. For convenience, the same secrecy rate constraint of TUs in different satellite beams is preset as $ Q $ for the simulation. Particularly, the impact of BS transmission power ($ P_B $), satellite transmission power ($ P_S $), number of BS transmit antennas ($ M $), and secrecy rate constraint of TUs ($ Q $) on the maximum secrecy rate performance of SU are evaluated as follows. 

\begin{figure}[h]
	\centering
	\includegraphics[width=0.47\textwidth]{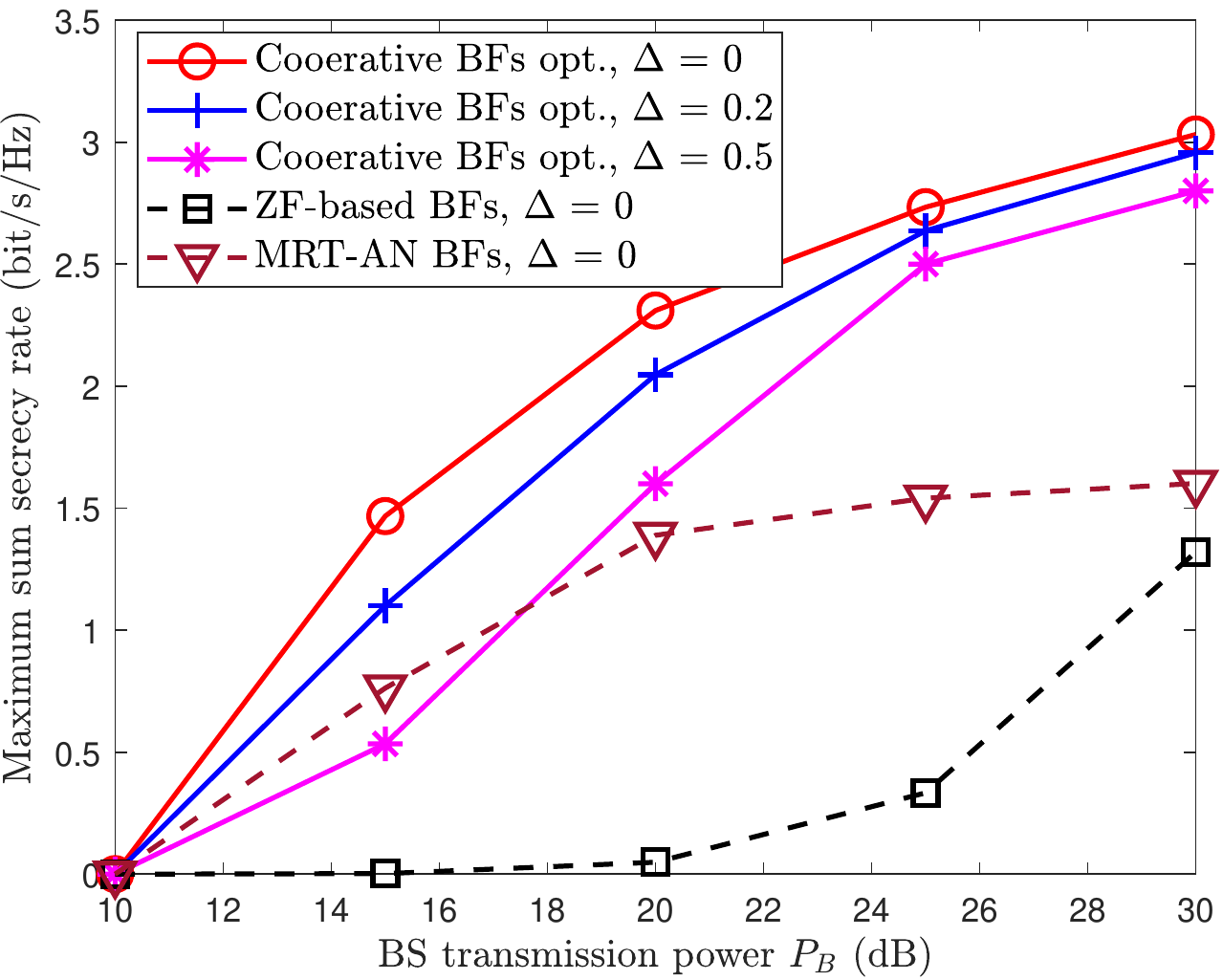}
	\caption{Secrecy rate performance of SUs vs. BS transmission power. (${P_S} = 20$ dB, $ Q = 0.5$ bit/s/Hz, $ M = 4 $.)}
	\label{PB}
\end{figure}
Fig. \ref{PB} shows the impact of BS transmission power on the maximum sum secrecy rate of SUs. From Fig. \ref{PB}, it can be seen that the maximum secrecy rate of SU is monotonically increasing of the BS transmission power. This is because the more green interference from BS downlink degrades the eavesdropping of satellite link as the BS transmission power increases. Particularly, it can be found that the objective value in (\ref{p3-1}) increases as the BS transmission power ($ {{\rm{Tr}}\left(  {{{\bf{F}}_i}}\right)} $) increases. According to the Proposition 1, the objective function in the problem ${\mathcal{P}2}$ in (\ref{p22-1}) has the same monotonicity as that in (\ref{p3-1}), which indicates the maximum sum secrecy rate of SUs increases as the BS transmission power. In addition, our proposed cooperative BFs optimization approach outperforms these two benchmarks, e.g., the approach that jointly optimizes both satellite BF and the PA of BS, and the approach that optimizes satellite BF with ZF-based BS BF. This is because the BS power focuses more on the main channel of TU by MRT-based BF while the AN damages the Eve poorly. Whereas the ZF-based BS BF approach doesn't concentrate the green interference and the AN from BS to damage the Eve, which improves the secrecy performance of SU slightly only by the satellite BF since the similarity of satellite channels.   

In Fig. \ref{SRvsPs}, the impact of satellite transmission power on the maximum sum secrecy rate of SUs is evaluated. 
From Fig. \ref{SRvsPs}, it can be seen that the maximum sum secrecy rate of SUs is monotonically increasing of the satellite transmission power. This is because the more green interference from satellite degrades the eavesdropping of TU as the satellite transmission power increases, thus the more BS power resource can serve as the green interfere for SU to degrade the Eve. Similarly, the objective value in (\ref{p3-1}) increases as the satellite transmission power ($ {{\rm{Tr}}\left(  {{{\bf{W}}_i}}\right)} $) increases. According to the Proposition 1, the maximum sum secrecy rate of SUs increases as the satellite transmission power. In addition, our proposed cooperative BFs optimization approach outperforms the approach that jointly optimizes both satellite BF and the PA of BS and the approach that optimizes satellite BF with ZF-based BS BF. Particularly, the approach that optimizes satellite BF with ZF-based BS BF outperforms the approach that jointly optimizes both satellite BF in low satellite transmission power region, which indicates that the main channel of TU is damaged slightly by satellite with low transmission power and the BS tends to suppress the Eve rather than to enhance the main channel of TU for guaranteeing the secrecy rate constraint.    

\begin{figure}[t]
	\centering
	\includegraphics[width=0.47\textwidth]{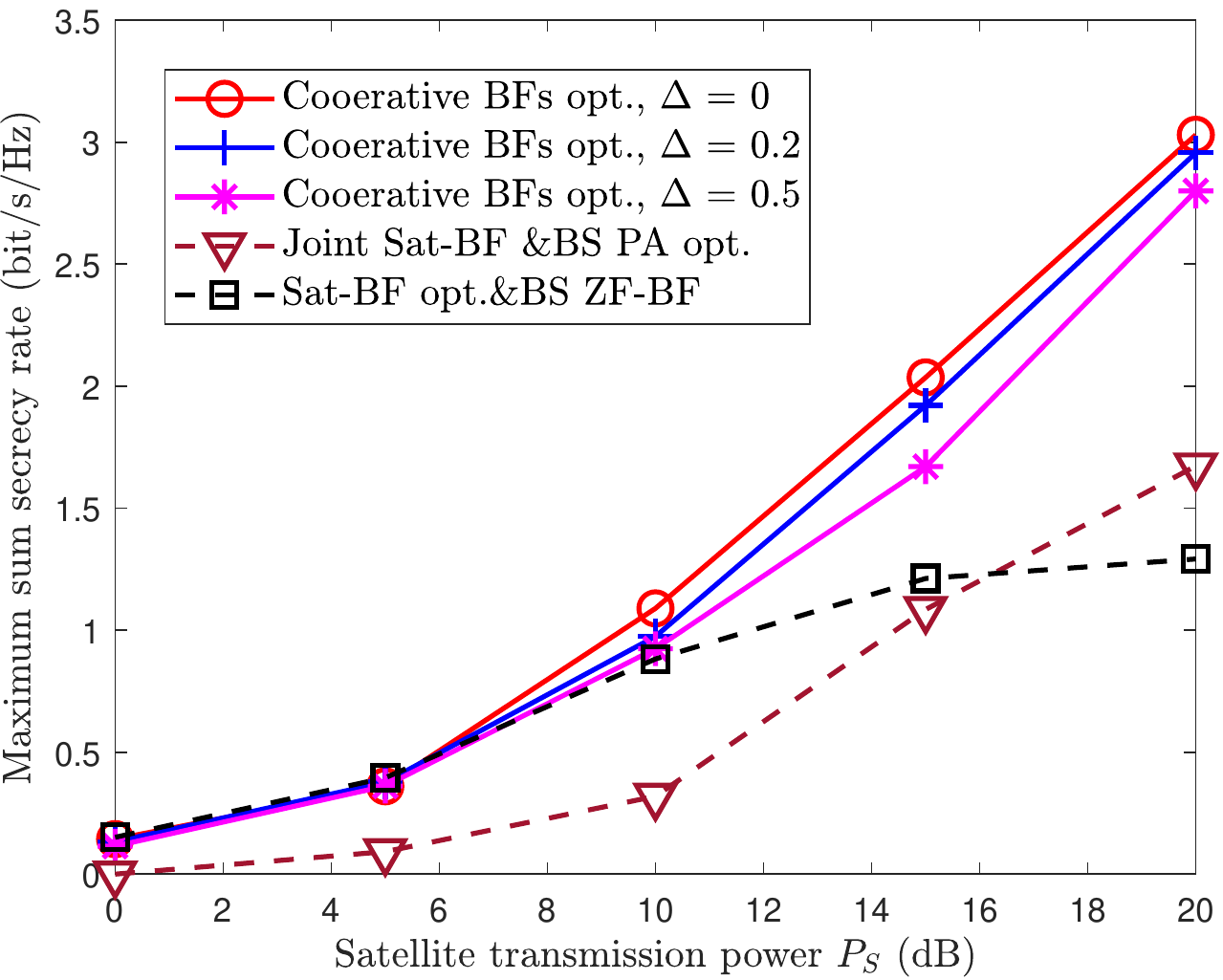}
	\caption{Secrecy rate performance of SU vs. satellite transmission power. (${P_B} = 30$ dB, $ Q = 0.5$ bit/s/Hz, $ M = 4 $.)}
	\label{SRvsPs}
\end{figure}

\begin{figure}[!h]
	\centering
	\includegraphics[width=0.47\textwidth]{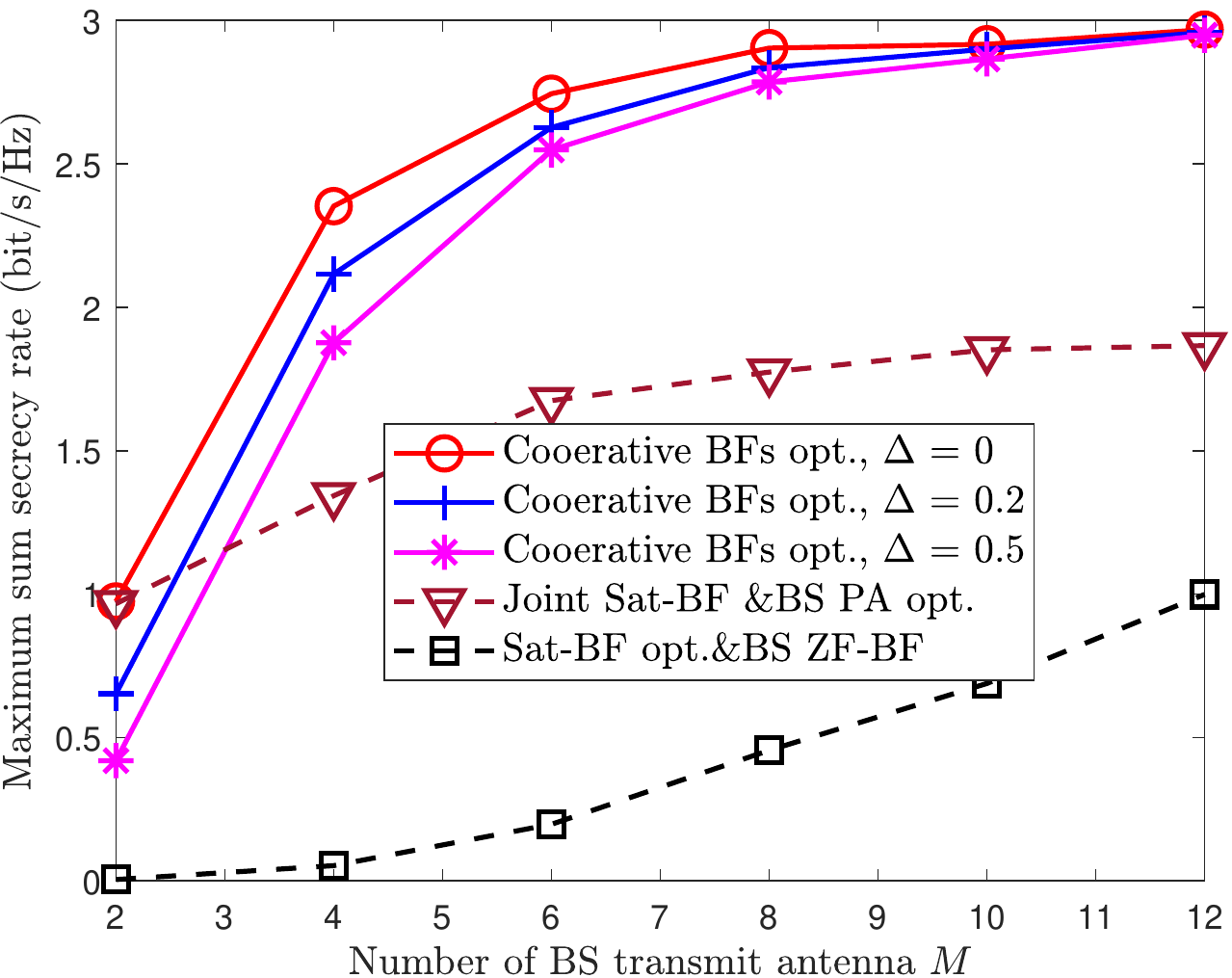}
	\caption{Secrecy rate performance of SU vs. the number of BS transmit antennas. (${P_B} = 30$ dB, ${P_S} = 20$ dB, $ Q = 0.5$ bit/s/Hz, $ M = 4 $.)}
	\label{SR_Antenna}
\end{figure}

\begin{figure}[!h]
	\centering
	\includegraphics[width=0.47\textwidth]{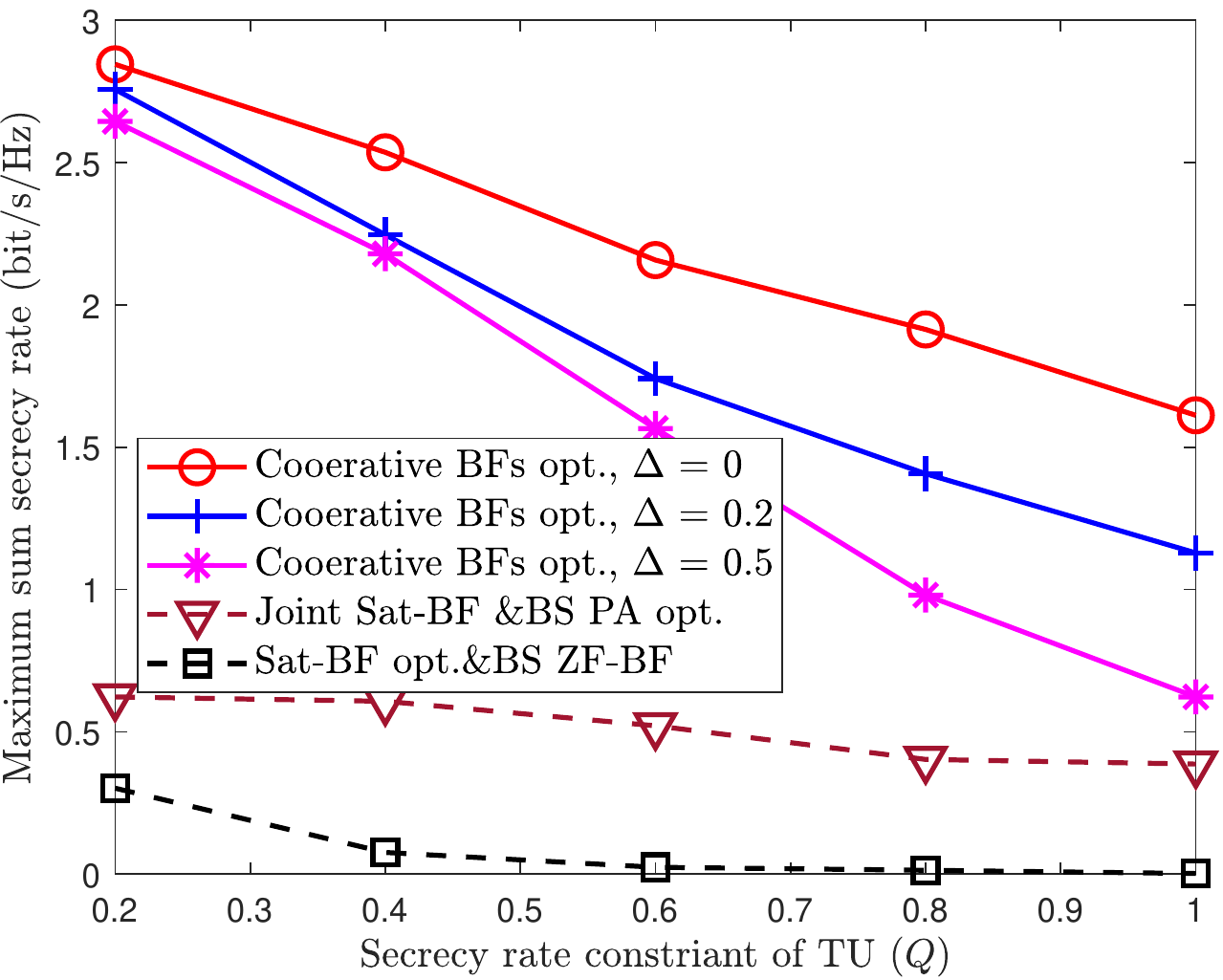}
	\caption{Secrecy rate performance of SU vs. secrecy rate constraint of TUs ($ Q $). (${P_S} = 20$ dB, ${P_B} = 20$ dB, $ M = 4 $.)}
	\label{QM}
\end{figure}

Fig. \ref{SR_Antenna} shows the impact of BS transmit antenna number on the maximum sum secrecy rate of SUs. From Fig. \ref{SR_Antenna}, it can be seen that the maximum sum secrecy rate of SUs is monotonically increasing of the number of BS transmit antennas. This is because the green interference from BS downlink can focus more centralized to degrade the eavesdropping of satellite link as the BS transmission power increases. Besides, it is observed that the impact of channel estimation error declines gradually as the number of BS transmit antenna increases, which happens since the antenna diversity could benefit from the channel estimation error. In addition, for the fixed BF schemes at the BS side, the signal or the AN can concentrate more on its direction as the number of BS transmit antenna increases. 

Fig. \ref{QM} shows the impact of secrecy rate constraint of TUs on the maximum sum secrecy rate of SU. From Fig. \ref{QM}, it can be seen that the maximum sum secrecy rate of SUs is monotonically decreasing of the secrecy rate constraint of TUs. 
This is because the BS transmission power focus more on the main channel of TU to guarantee the secrecy rate constraint of TUs as the constraint increases, while the green interference from the BS for damaging the Eve is degraded. 
Also, it can be seen our proposed cooperative BFs optimization approach outperforms the approach that jointly optimizes both satellite BF and the PA of BS and the approach that optimizes satellite BF with ZF-based BS BF, which verifies the efficiency of our proposed approach.    
%

\section{Conclusion}
In this paper, a symbiotic secure transmission scheme based on cooperative BF optimization in integrated satellite-terrestrial communications has been proposed. Particularly, the co-channel interference induced by spectrum sharing within satellite-terrestrial networks and the inter-beam interference due to frequency reuse among satellite beams serve as the green interference has been used to ensure secure transmissions of both satellite and terrestrial links simultaneously. Specifically, the problem to maximize the sum secrecy rate of SUs is formulated and the BFs optimizing is conducted cooperatively, where the secrecy rate constraint of each TU is guaranteed. Furthermore, the Taylor expansion and SDR have been adopted to reformulate this problem, and a SCA based joint satellite-terrestrial BFs optimization approach has been proposed to solve this problem. The tightness of the relaxation has also been proved. In addition, numerical results have verified the efficiency of our proposed cooperative BFs optimization approach and revealed that the inherent green interference from internal system without additional assistance can assist the implement of symbiotic secure transmissions in integrated satellite-terrestrial networks. 
\appendices
\section{Proof of Theorem 1}
\begin{proof}
	The Lagrangian function of ${\mathcal{P}2}$ can be obtained as shown at the bottom of this page.
		\setcounter{equation}{55} 
	\begin{figure*}[b]
	\hrulefill
	\begin{equation}\label{Eq-Lagrang}
		\begin{array}{*{20}{l}}
			{{\cal L}\left( {\left\{ {{\xi _k},{\varsigma _k},{\lambda _k},{\rho _k},{\tau _k},{\zeta _k},{\upsilon _k},{\vartheta _k},{\phi _k},{\beta _k},{\kappa _k},{{\bf{W}}_k},{{\bf{F}}_k},{{\bf{U}}_k},{{\bf{L}}_k}} \right\}} \right) = }\\
			{\sum\limits_{i = 1}^N {{\rm{Tr}}\left( {{{\bf{W}}_i} + {{\bf{F}}_i}} \right)}  - \sum\limits_{k = 1}^N {{\xi _k}\left( {{s_k} - {\mu _k} - {q_k} + {v_k} - {\varphi^\circ }} \right)}  + \sum\limits_{k = 1}^N {{\lambda _k}\left( {{e^{{s_k}}} - \sum\limits_{i = 1}^N {{\rm{Tr}}\left( {{{\bf{H}}_{su,i}}{{\bf{W}}_i}} \right)}  - {\rm{Tr}}\left( {{{\bf{G}}_{su,k}}{{\bf{F}}_k}} \right) - 1} \right)} }\\
			{ + \sum\limits_{k = 1}^N {{\rho _k}\left( {{e^{{v_k}}} - \sum\limits_{i \ne k}^N {{\rm{Tr}}\left( {{{\bf{H}}_{e,i}}{{\bf{W}}_i}} \right)}  - {\rm{Tr}}\left( {{{\bf{G}}_{e,k}}{{\bf{F}}_k}} \right) - 1} \right)}  + \sum\limits_{k = 1}^N {{\tau _k}\left( {\sum\limits_{i \ne k}^N {{\rm{Tr}}\left( {{{\bf{H}}_{su,i}}{{\bf{W}}_i}} \right)}  + {\rm{Tr}}\left( {{{\bf{G}}_{su,k}}{{\bf{F}}_k}} \right) - {e^{{{\tilde \mu }_k}}}\left( {{\mu _k} - {{\tilde \mu }_k}} \right)} \right)} }\\
			{ + \sum\limits_{k = 1}^N {{\zeta _k}\left( {\sum\limits_{i = 1}^N {{\rm{Tr}}\left( {{{\bf{H}}_{e,i}}{{\bf{W}}_i}} \right)}  + {\rm{Tr}}\left( {{{\bf{G}}_{e,k}}{{\bf{F}}_k}} \right) - {e^{{{\tilde q}_k}}}\left( {{q_k} - {{\tilde q}_k}} \right)} \right)}  + \sum\limits_{k = 1}^N {{\upsilon _k}\left( {\sum\limits_{i = 1}^N {{\rm{Tr}}\left( {{{\bf{H}}_{tu,i}}{{\bf{W}}_i}} \right)}  - {e^{{{\tilde \eta }_k}}}\left( {{\eta _k} - {{\tilde \eta }_k}} \right)} \right)} }\\
			{ + \sum\limits_{k = 1}^N {{\vartheta _k}\left( {{e^{{\tau _k}}} - \sum\limits_i^N {{\rm{Tr}}\left( {{{\bf{H}}_{tu,i}}{{\bf{W}}_i}} \right)} {\rm{ - Tr}}\left( {{{\bf{G}}_{tu,k}}{{\bf{F}}_k}} \right) - 1} \right)}  + \sum\limits_{k = 1}^N {{\phi _k}\left( {{\eta _k} + {q_k} - {\tau _k} - {\alpha _k} + \frac{{{Q_{tu,k}}}}{{{{\log }_2}e}}} \right) - \sum\limits_{k = 1}^N {{{\bf{U}}_k}{{\bf{W}}_k}} } }\\
			{ + \sum\limits_{k = 1}^N {{\beta _k}\left( {\sum\limits_{k = 1}^N {{\rm{Tr}}\left( {{{\bf{W}}_k}} \right)}  - {P_S}} \right)}  + \sum\limits_{k = 1}^N {{\varsigma _k}\left( {{e^{{\alpha _k}}} - \sum\limits_{i = 1}^N {{\rm{Tr}}\left( {{{\bf{H}}_{e,i}}{{\bf{W}}_i}} \right)}  - 1} \right) + \sum\limits_{k = 1}^N {{\kappa _k}\left( {{\rm{Tr}}\left( {{{\bf{F}}_k}} \right) - {P_B}} \right)}  - \sum\limits_{k = 1}^N {{{\bf{L}}_k}{{\bf{F}}_k}} } }.
		\end{array}
	\end{equation}
\end{figure*}
	
	Based on (\ref{Eq-Lagrang}), we take the partial derivative of $ \mathcal L\left(  \cdot  \right) $ with respect to $ {{{\bf{W}}_k}} $ and apply KKT conditions as follows
	\begin{equation}\label{KKT1}
		\begin{split}
			&{\bf{D}} - {{\bf{H}}_{su,k}}\sum\limits_{k = 1}^N {{\lambda _k}}  - {{\bf{U}}_k} = {\bf{0}},
		\end{split}
	\end{equation}
	\begin{equation}\label{KKT2}
		{{\bf{U}}_k}{{\bf{W}}_k} = {\bf{0}},
	\end{equation}
	\begin{equation}\label{KKT3}
		{{\bf{W}}_k} \succeq {\bf{0}},
	\end{equation}
	where 
	\begin{align}\label{Eq40}
		{\bf{D}} = &{{\bf{I}}_N}\big( {1 + \sum\limits_{i = 1}^N {{\beta _i}} } \big) + {{\bf{H}}_{e,k}}\big( {\sum\limits_{i = 1}^N {{\zeta _i}}  - \sum\limits_{i \ne k}^N {{\rho _i}}  - \sum\limits_{i = 1}^N {{\varsigma _i}} } \big) \nonumber\\
		&+{{\bf{H}}_{tu,k}}\big( {\sum\limits_{i = 1}^N {{\upsilon _i}}  - \sum\limits_{i = 1}^N {{\vartheta _i}} } \big) + {{\bf{H}}_{su,k}}\sum\limits_{i \ne k}^N {{\tau _i}} .
	\end{align}
	By using (\ref{KKT1}--\ref{KKT2}), we have the reformulation as follows
	\begin{equation}\label{Eq41}
		{\bf{D}}{{\bf{W}}_k} = {{\bf{H}}_{su,k}}{{\bf{W}}_k}\sum\limits_{k = 1}^N {{\lambda _k}}.
	\end{equation}
	
	From (\ref{Eq40}) and (\ref{Eq41}), we can find that 
	\begin{equation}\label{key}
		N - 2 \le rank\left( {\bf{D}} \right) \le N,
	\end{equation}
	
	By (\ref{KKT1}), we have $ {{\bf{E}} = {\bf{D}}} $ by denoting $ {\bf{E}} = {{\bf{H}}_{su,k}}\sum\limits_{k = 1}^N {{\lambda _k}}  + {{\bf{U}}_k} $, and thus 
	$ rank\left( {\bf{E}} \right) \ge N - 2 $. With (\ref{KKT2}), we have 
	\begin{align}\label{key}
		rank\left( {{\bf{E}}{{\bf{W}}_k}} \right) &= 
		rank\big( {{{\bf{H}}_{su,k}}{{\bf{W}}_k}\sum\limits_{k = 1}^N {{\lambda _k}} } \big) \nonumber\\
		&\le rank\left( {{{\bf{H}}_{su,k}}} \right) = 1
	\end{align}
	and 
$ 		rank\left( {\bf{E}} \right) + rank\left( {{{\bf{W}}_k}} \right) \le rank\left( {{\bf{E}}{{\bf{W}}_k}} \right) + N. $

	If $ rank\left( {{\bf{E}}{{\bf{W}}_k}} \right) =1$, we have $ rank\left( {{{\bf{W}}_k}} \right) = N $ and $ rank\left( {\bf{E}} \right) \le 1 $, which conflicts with the precondition $ rank\left( {\bf{E}} \right) \ge N - 2 $. Consequently, the provision of $ rank\left( {{\bf{E}}{{\bf{W}}_k}} \right) =0$ should be satisfied. Then, $ rank\left( {{{\bf{W}}_k}} \right) \le 2 $ is achieved.
	
	Further, we reformulated (\ref{KKT1}) as 
	\begin{equation}\label{Eq45}
		{\bf{\Theta }} - {\bf{E}} - {{\bf{H}}_{e,k}}\big( {\sum\limits_{i \ne k}^N {{\rho _i}}  + \sum\limits_{i = 1}^N {{\varsigma _i}} } \big) - \sum\limits_{i = 1}^N {{\vartheta _i}{{\bf{H}}_{tu,k}}}  = {\bf{0}},
	\end{equation}
	where $ {\bf{\Theta }} = {{\bf{I}}_N}\big( {1 + \sum\limits_{i = 1}^N {{\beta _i}} } \big) + {{\bf{H}}_{e,k}}\sum\limits_{i = 1}^N {{\zeta _i}}  + \sum\limits_{i \ne k}^N {{\tau _i}{{\bf{H}}_{su,k}}}  + \sum\limits_{i = 1}^N {{\upsilon _i}{{\bf{H}}_{tu,k}}}  $, and it can be  observed that $ {\bf{\Theta }} \succ {\bf{0}} $ due to $ {{\beta _i}} \geq0,{{\zeta _i}}\geq0, {\tau _i}\geq0$, and ${\upsilon _i}\geq0$.
	
	By post-multiplying $ {{{\bf{W}}_k}} $ at both sides of (\ref{Eq45}) and using $ rank\left( {{\bf{E}}{{\bf{W}}_k}} \right) =0$, (\ref{Eq45}) can be rewritten as	
	\begin{equation}\label{Eq46}
		{\bf{\Theta }}{{\bf{W}}_k} = \big( {{\bf{E}} + {{\bf{H}}_{e,k}}(\sum\limits_{i \ne k}^N {{\rho _i}}  + \sum\limits_{i = 1}^N {{\varsigma _i}} ) + \sum\limits_{i = 1}^N {{\vartheta _i}{{\bf{H}}_{tu,k}}} } \big){{\bf{W}}_k}.
	\end{equation}

	In (\ref{Eq46}), since $ {\bf{\Theta }} \succ {\bf{0}} $, we have 
	\begin{equation}\label{Eq47}
		rank\left( {{\bf{\Theta }}{{\bf{W}}_k}} \right) = rank\left( {{{\bf{W}}_k}} \right),
	\end{equation}
	and thus (\ref{Eq-68}) is achieved.
	\begin{align}\label{Eq-68}
	&rank\big( {{{\bf{W}}_k}} \big) = rank\big( {\big( {{{\bf{H}}_{e,k}}(\sum\limits_{i \ne k}^N {{\rho _i}}  + \sum\limits_{i = 1}^N {{\varsigma _i}} ) + \sum\limits_{i = 1}^N {{\vartheta _i}{{\bf{H}}_{tu,k}}} } \big){{\bf{W}}_k}} \big) \nonumber\\
	&\le rank\big( {{{\bf{H}}_{e,k}}{{\bf{W}}_k}(\sum\limits_{i \ne k}^N {{\rho _i}}  + \sum\limits_{i = 1}^N {{\varsigma _i}} )} \big) + rank\big( {{{\bf{H}}_{tu,k}}{{\bf{W}}_k}\sum\limits_{i = 1}^N {{\vartheta _i}} } \big).
\end{align}

	Particularly, 
	\begin{align}\label{Eq49}
			rank\big( {{{\bf{H}}_{e,k}}{{\bf{W}}_k}(\sum\limits_{i \ne k}^N {{\rho _i}}  + \sum\limits_{i = 1}^N {{\varsigma _i}} )} \big) \le rank\left( {{{\bf{H}}_{e,k}}} \right) = 1,\\
			rank\big( {{{\bf{H}}_{tu,k}}{{\bf{W}}_k}\sum\limits_{i = 1}^N {{\vartheta _i}} } \big) \le rank\left( {{{\bf{H}}_{tu,k}}} \right) = 1.
	\end{align}	
	However, if $ rank\left( {{{\bf{W}}_k}} \right) = 2 $, it can be observed from (\ref{Eq49}) that $ rank\big( {{{\bf{H}}_{e,k}}{{\bf{W}}_k}(\sum\limits_{i \ne k}^N {{\rho _i}}  + \sum\limits_{i = 1}^N {{\varsigma _i}} )} \big) =1 $ and $ rank\big( {{{\bf{H}}_{tu,k}}{{\bf{W}}_k}\sum\limits_{i = 1}^N {{\vartheta _i}} } \big) =1$, which indicates an incompatible result that $ rank\big( {{{\bf{W}}_k}} \big) = N $.
	
	Therefore, $ rank\left( {{{\bf{W}}_k}} \right) \le 1 $ is kept, and $ rank\left( {{{\bf{W}}_k}} \right) = 0  $ can not be a solution and which should be discarded. Finally, $ rank\left( {{{\bf{W}}_k}} \right) = 1 $ is proved.
	
	Similarly, for the proof of $ Rank\left( {{{\bf{F}}_k}} \right) = 1 $, we take the partial derivative of $ \mathcal L\left(  \cdot  \right) $ with respect to $ {{{\bf{F}}_k}} $ and apply KKT conditions as follows
	\begin{equation}\label{Eq-50}
		{\bf{A}}- {\vartheta _k}{{\bf{G}}_{tu,k}} - {{\bf{L}}_k} = {\bf{0}},
	\end{equation}
	\begin{equation}\label{Eq-51}
		{{\bf{L}}_k}{{\bf{F}}_k} = {\bf{0}},
	\end{equation}
	\begin{equation}\label{Eq-52}
		{{\bf{L}}_k} \succeq {\bf{0}},
	\end{equation}
	where
	\begin{equation}
		{\bf{A}} = (1 + \sum\limits_{k = 1}^N {{\kappa _k}} ){{\bf{I}}_N} + \left( {{\tau _k} - {\lambda _k}} \right){{\bf{G}}_{su,k}} + \left( {{\zeta _k} - {\rho _k}} \right){{\bf{G}}_{e,k}},
	\end{equation}
	and it is observed that $ N - 2 \le rank\left( {\bf{A}} \right) \le N $. Due to 
	$ {{\bf{G}}_{tu,k}} = {\bf{h}}_{tu,k}^\dag {\bf{h}}_{tu,k}^{} $, $ rank\left( {{{\bf{G}}_{tu,k}}} \right) = 1 $. Thus, similar to the way of proof for $ Rank\left( {{{\bf{W}}_k}} \right) = 1 $ above, $ Rank\left( {{{\bf{F}}_k}} \right) = 1 $ can also be proved.  
\end{proof}
\bibliographystyle{IEEEtran}
\bibliography{IEEEabrv,zsyin_Sub1}
\begin{IEEEbiography}[{\includegraphics[width=1in,height=1.25in,clip,keepaspectratio]{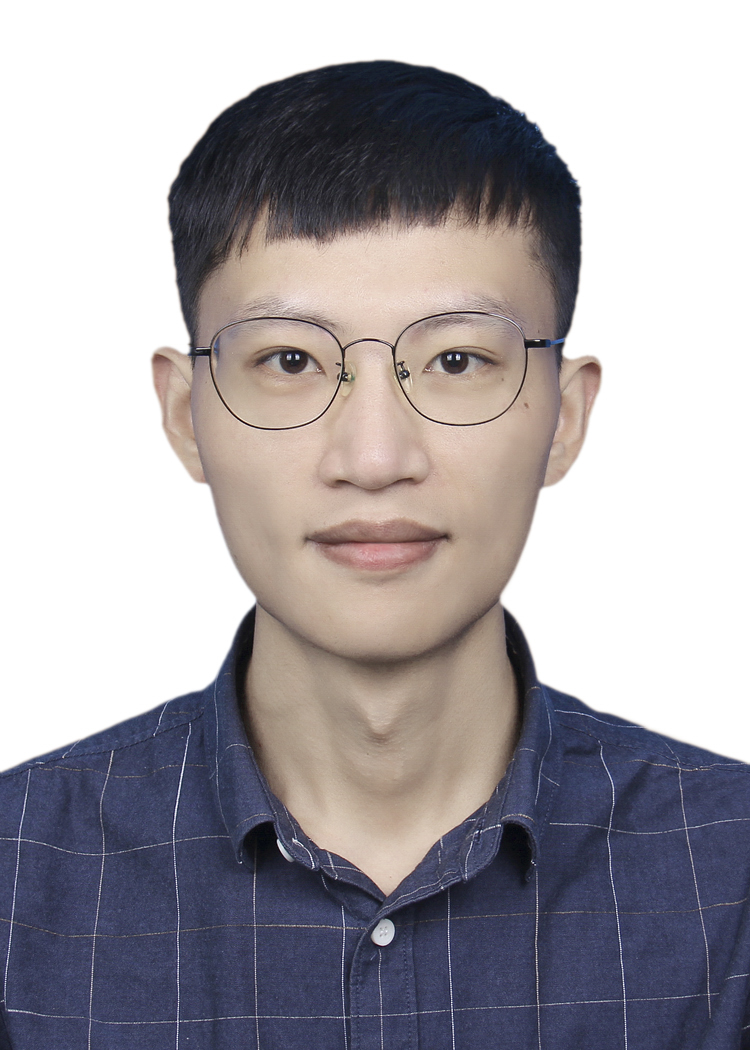}}]{Zhisheng Yin}
	(M'20) received his Ph.D. degree from the School of Electronics and Information Engineering, Harbin Institute of Technology, Harbin, China, in 2020, and the B.E. degree from the Wuhan Institute of Technology, the B.B.A. degree from the Zhongnan University of Economics and Law, Wuhan, China, in 2012, and the M.Sc. degree from the Civil Aviation University of China, Tianjin, China, in 2016. From Sept. 2018 to Sept. 2019, Dr. Yin visited in
	BBCR Group, Department of Electrical and Computer Engineering, University
	of Waterloo, Canada. He is currently an Associate Professor with School of Cyber Engineering, Xidian University, Xi'an, China. His research interests include space-air-ground integrated networks, wireless communications, digital twin, and physical layer security.
\end{IEEEbiography}
\begin{IEEEbiography}[{\includegraphics[width=1in,height=1.25in,clip,keepaspectratio]{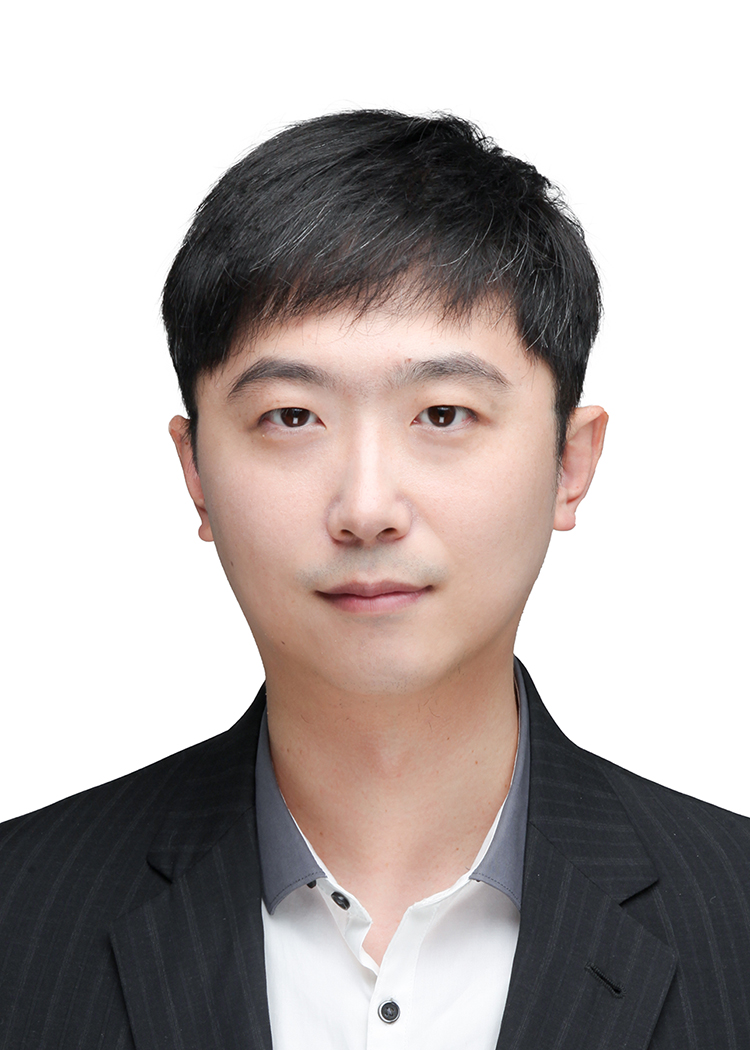}}]{Nan Cheng}
	(M'16) received the Ph.D. degree from the Department of Electrical and Computer Engineering, University of Waterloo in 2016, and B.E. degree and the M.S. degree from the Department of Electronics and Information Engineering, Tongji University, Shanghai, China, in 2009 and 2012, respectively. He worked as a Post-doctoral fellow with the Department of Electrical and Computer Engineering, University of Toronto, from 2017 to 2019. He is currently a professor with State Key Lab. of ISN and with School of Telecommunications Engineering, Xidian University, Shaanxi, China.  His current research focuses on B5G/6G, space-air-ground integrated network, big data in vehicular networks, and self-driving system. His research interests also include performance analysis, MAC, opportunistic communication, and application of AI for vehicular networks.
\end{IEEEbiography}
\begin{IEEEbiography}[{\includegraphics[width=1in,height=1.25in,clip,keepaspectratio]{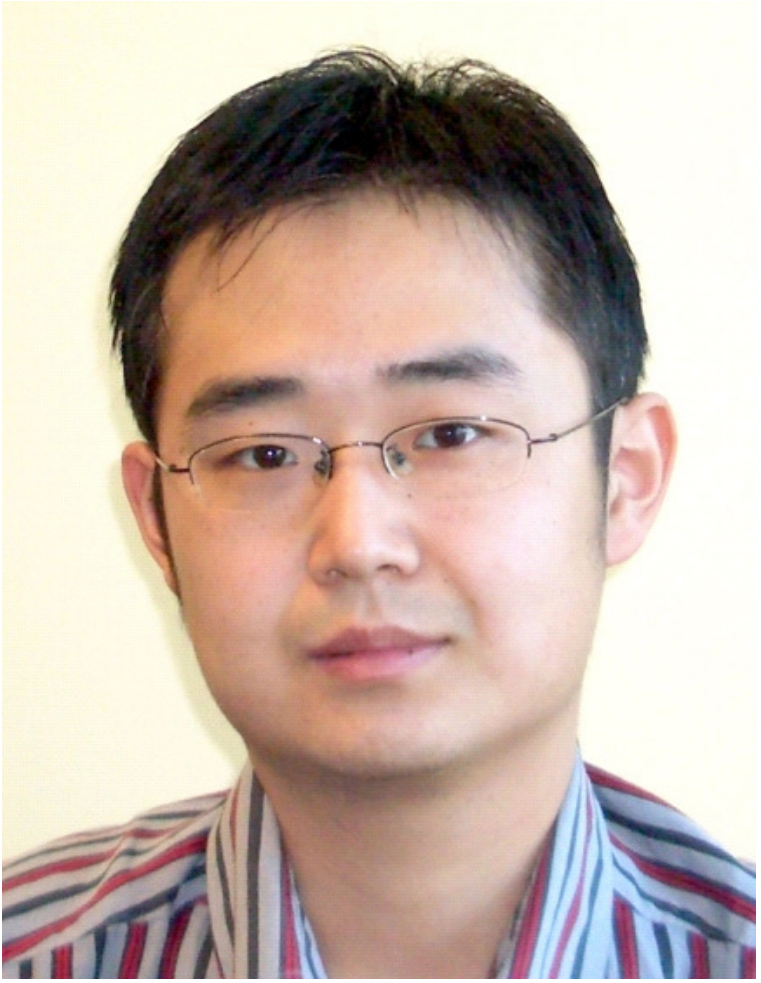}}]{Tom H. Luan}
	(M'14, SM'17) received the B.Eng. degree from Xi’an Jiaotong University, Xi’an, China, in 2004, the M.Phil. degree from The Hong Kong University of Science and Technology, Hong Kong, in 2007, and the Ph.D. degree from the University of Waterloo, Waterloo, ON, Canada, in 2012. He is currently a Professor with the School of Cyber Engineering, Xidian University, Xi’an. He has authored or coauthored more than 40 journal articles and 30 technical articles in conference proceedings. His research interests include content distribution and media streaming in vehicular ad hoc networks, peer-to-peer networking, and the protocol design and performance evaluation of wireless cloud computing and edge computing. Dr. Luan was the recipient of one U.S. patent. He was a TPC Member of the IEEE Global Communications Conference, the IEEE International Conference on Communications, and the IEEE International Symposium on Personal, Indoor and Mobile Radio Communications, and the Technical Reviewer for multiple IEEE Transactions, including the IEEE Transactions on Mobile Computing, the IEEE Transactions on Parallel and Distributed Systems, the IEEE Transactions on Vehicular Technology, the IEEE Transactions on Wireless Communications, and the IEEE Transactions on Intelligent Transportation Systems.
\end{IEEEbiography}
\begin{IEEEbiography}[{\includegraphics[width=1in,height=1.25in,clip,keepaspectratio]{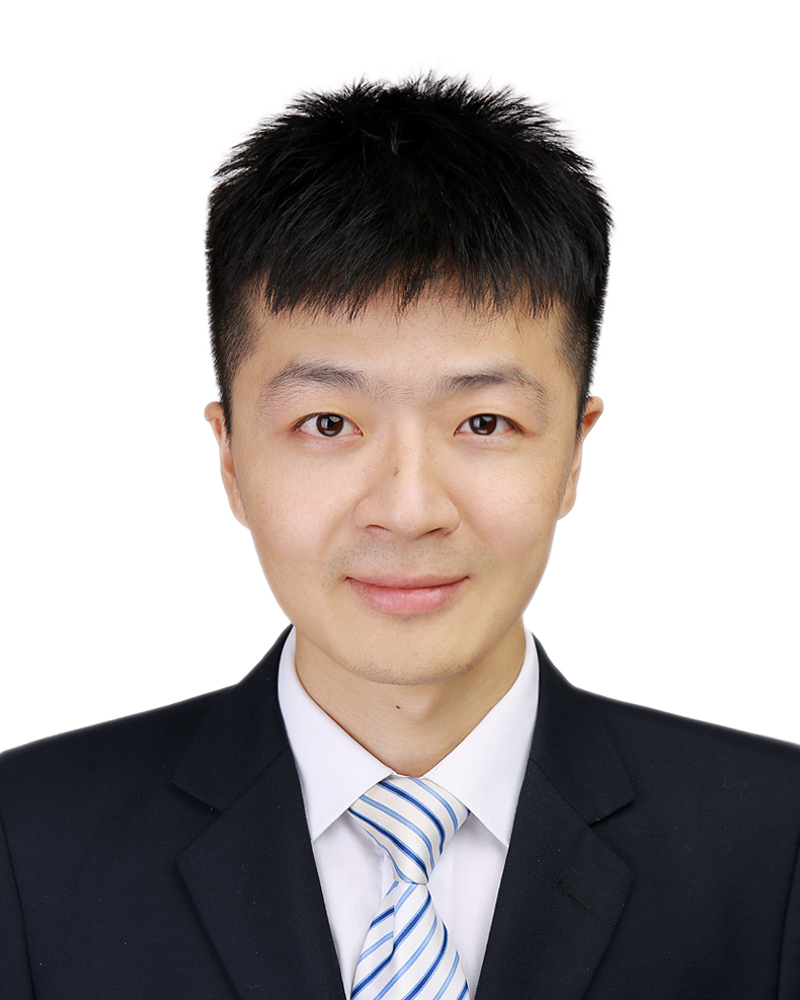}}]{Yilong Hui}
	(M'18) received the Ph.D. degree in control
	theory and control engineering from Shanghai University,
	Shanghai, China, in 2018. He is currently
	an Associate Professor with the State Key Laboratory of Integrated
	Services Networks, and with the School of
	Telecommunications Engineering, Xidian University,
	China. He has published over 50 scientific articles
	in leading journals and international conferences.
	His research interests include wireless
	communication, mobile edge computing, vehicular networks, intelligent transportation systems and autonomous driving. He was the recipient of the Best Paper Award of International Conference WiCon2016 and IEEE
	Cyber-SciTech2017.
\end{IEEEbiography}
\begin{IEEEbiography}[{\includegraphics[width=1in,height=1.25in,clip,keepaspectratio]{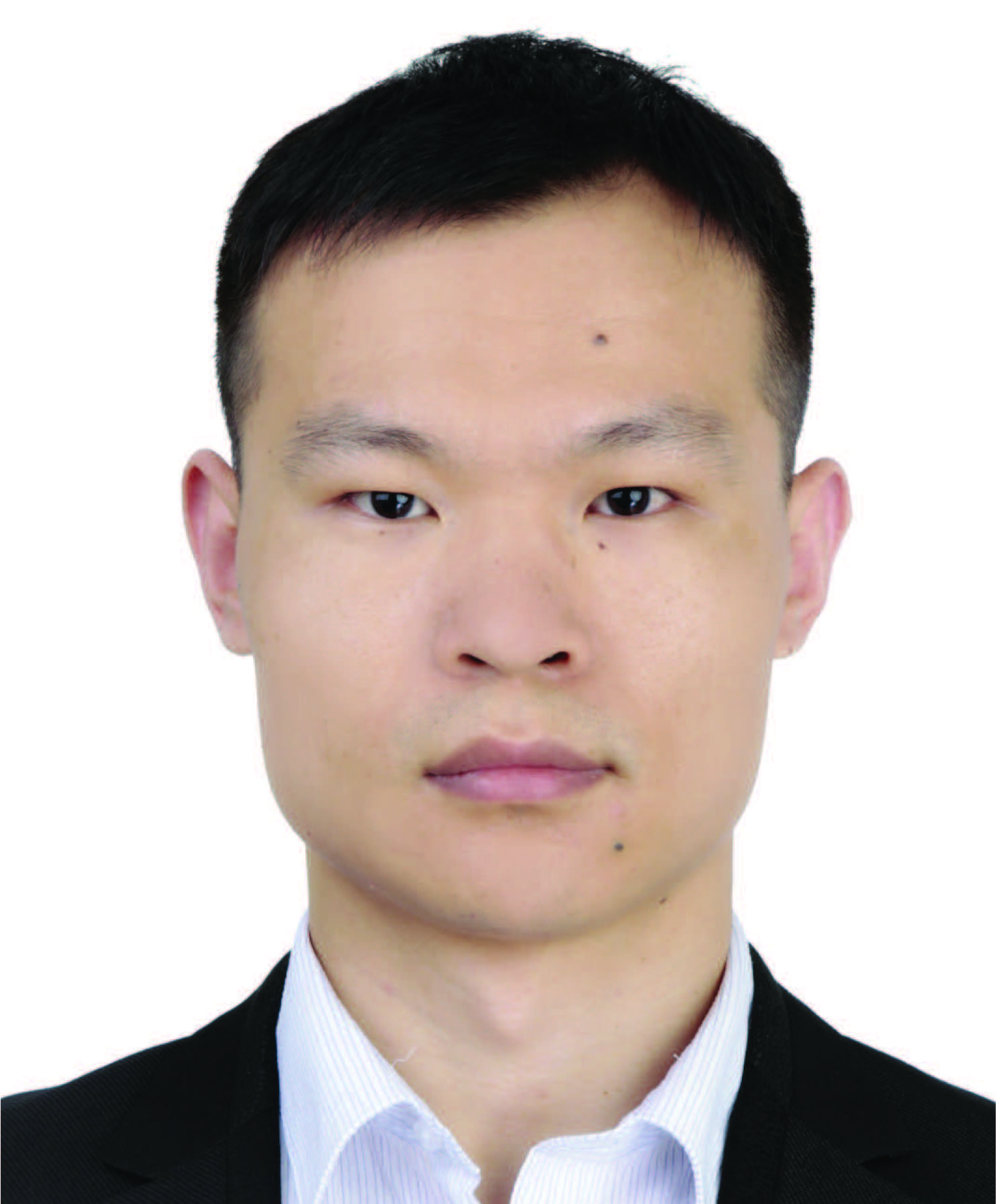}}]{Wei Wang}
	(M'19) received the B.Eng. degree in Information Countermeasure Technology and the M.Eng. degree in Signal and Information Processing from Xidian University in 2011 and 2014, respectively, and the PhD degree in Electrical and Electronic Engineering from Nanyang Technological University (NTU), Singapore, in 2018. From Sep. 2018 to Aug. 2019, he was a postdoctoral fellow at the Department of Electrical and Computer Engineering, University of Waterloo, Canada.  Currently, he is a Professor at Nanjing University of Aeronautics and Astronautics. His research interests include wireless communications, space-air-ground integrated networks, wireless security, and blockchain. He was awarded the Chinese government award for outstanding self-financed students abroad in 2018, and the Young Elite Scientist Sponsorship Program, China Association for Science and Technology in 2021.
\end{IEEEbiography}
\end{document}